\documentclass[11pt]{article}
\usepackage[letterpaper,margin=1in]{geometry}
\usepackage{amsmath, amsthm, mathtools, amssymb}
\usepackage[compress]{natbib}
\usepackage{booktabs}
\usepackage{enumitem}
\usepackage{microtype}
\usepackage{xcolor}
\usepackage{hyperref}
\usepackage{caption}
\usepackage{titlesec}
\usepackage[capitalise]{cleveref}
\usepackage{float}
\usepackage{tikz}
\usepackage{pgfplots}
\usepackage{todonotes}
\pgfplotsset{compat=1.18}

\hypersetup{
    colorlinks=true,
    linkcolor=red!70!black,
    citecolor=blue,
    urlcolor=red!70!black
}

\newtheorem{theorem}{Theorem}[section]
\newtheorem{lemma}[theorem]{Lemma}
\newtheorem{proposition}[theorem]{Proposition}
\newtheorem{corollary}[theorem]{Corollary}
\newtheorem{definition}[theorem]{Definition}
\newtheorem{example}[theorem]{Example}
\newtheorem{remark}[theorem]{Remark}

\newcommand{\E}{\mathbb{E}}
\newcommand{\Rplus}{\mathbb{R}_{\ge 0}}
\newcommand{\F}{\mathcal{F}}
\newcommand{\Ftwo}{\mathcal{F}^{(2)}}
\newcommand{\A}{\mathcal{A}}
\newcommand{\M}{\mathcal{M}}
\newcommand{\Rev}{\operatorname{Rev}}
\newcommand{\OPT}{\operatorname{OPT}}
\newcommand{\rev}{\operatorname{rev}}
\newcommand{\ind}{\mathbf{1}}
\DeclareMathOperator{\Pdim}{Pdim}

\title{Auctions with Price Predictions}

\author{Siddharth Prasad \\ \small Toyota Technological Institute at Chicago \and Dravyansh Sharma\\ \small Toyota Technological Institute at Chicago \and Alec Sun\\ \small University of Chicago \and Muthu Sundar \footnote{Authors are ordered alphabetically.} \\ \small University of Chicago}

\date{}

\begin{document}

\maketitle

\begin{abstract}
    We design auctions for the sale of a single item with unlimited supply given a single prediction of the revenue-maximizing uniform price. This departs from prior work on auctions with predictions which typically assumes predictions of every bidder's value. Our main result is a characterization of the Pareto frontier for consistency and robustness attainable by any universally truthful auction. We show that a mechanism that randomizes between posting the predicted price and conducting an optimal prior-free fallback auction is Pareto-optimal. We then extend our results to achieve graceful degradation of revenue as a function of the prediction accuracy. Finally, we study how such a prediction can be obtained from historical market data through the lens of learning theory. Together, our results give an end-to-end account of how a price prediction can be learned and used robustly to maximize auction revenue.
\end{abstract}
\section{Introduction} \label{sec:introduction}

Setting revenue-maximizing prices in auctions is one of the most important challenges in the theory and practice of mechanism design. \citet{myerson1981optimal} derived the revenue-maximizing auction for the sale of a single item among multiple bidders. Since then, revenue-optimal auction design has eluded analytical characterizations beyond very specific settings, for example optimal auctions for two to six items with a single bidder \citep{manelli2006bundling,pavlov2011optimal,giannakopoulos2018duality}. Revenue-optimal auction design for the sale of two items among multiple bidders remains an open problem.

In this paper, we design high-revenue pricing mechanisms for auctions with an unlimited supply of identical items and unit-demand bidders. This kind of sales mechanism is often referred to as a digital goods auction in reference to the fact that the marginal cost of production of additional units is zero. Our model captures reproducible goods and access rights such as software licenses, e-books, streaming subscriptions, and API access/cloud services, but more generally extends to markets where available inventory is effectively unlimited over the relevant demand range, for example rice, wheat, and water.

Our mechanisms receive as input a prediction of the revenue-maximizing uniform price, also known as the optimal monopoly price. Price predictions reflect the reality of data-driven demand forecasting and price optimization models used to drive revenue in massive industries like modern advertising markets, where reserve prices are routinely optimized from data. For example, a large field experiment on sponsored-search auctions at Yahoo found that theory-guided reserve prices derived from empirical bidder data substantially increased revenue \citep{OstrovskySchwarz2023}. Other example applications of data-driven price models for revenue optimization in electronic commerce include fashion retail \citep{caro2012clearance,ferreira2016analytics}, subscriptions for online job-recruiting service ZipRecruiter \citep{dube2023personalized}, and dynamic pricing for the Alibaba platform \citep{liu2019dynamic}.

Price predictions are a departure from the typical prediction model studied in the large body of work on {\em mechanism design with predictions} \citep{balkanski2024mechanism}. Much of that literature assumes access to point predictions about each bidder's value. Such bidder-level predictions, which we call \emph{value} predictions, are substantially more demanding to acquire, requiring bidder-specific data, bidder identification across markets, and high-dimensional value predictions that grow quickly in size with the market. These requirements are especially restrictive when bidders are new, anonymous, or observed only sparsely. Our informational model is more compressed, giving the seller one scalar prediction of a revenue-maximizing uniform price for the realized market. A price prediction of this form can be learned directly from historical markets, public features, or domain experts. We later show that the sample complexity of learning this shared price is only logarithmic in the number of bidders.

In this paper, we take the position that price predictions are a significantly more realistic kind of population-level prediction than bidder-wise value predictions. This position is bolstered by the large body of work on learning prices from data \citep{mohri2015revenue,mohri2016learning,shen2019learning,huchette2020contextual,huang2022learning,hu2025learning}. Beyond prices, deep-learning based auction design has obtained high-revenue auctions \citep{dutting2019optimal,wang2024Gemnet,wang2026bundleflow} by learning model weights to represent the allocation and pricing rules of the auction. And, on the theoretical front, a large body of work studies the sample complexity of learning auction parameters from data \citep{ColeRoughgarden2014,MorgensternRoughgarden2015,balcan2020efficient,balcan2021learning,BalcanSandholmVitercik2023}. In the aforementioned research strands, the pricing rule of the auction is learned directly from samples. Therefore, our price-prediction-augmented auctions framework may be viewed as an end-to-end pipeline from learning prices to an auction that is meaningfully able to take advantage of the predicted price. Such an end-to-end framework is largely missing from the existing literature on mechanism design with predictions. Our work addresses this gap with a more realistic model of predictions in auctions.

While price predictions are easier to obtain than value predictions, the mechanism design problem introduces new challenges. A predicted valuation profile reveals how demand is distributed across all bidders, and a single price suppresses nearly all of that information. Another major issue is the asymmetric nature of pricing errors. If the prediction is slightly below an optimal price, many profitable sales may survive, whereas a prediction just above the highest value can be rejected by every bidder and earn no revenue. The auction must therefore decide how much to trust the prediction while retaining meaningful revenue protection when it is wrong. Our main result tightly characterizes what is achievable in this setting.

\begin{theorem}[Informal main result] \label{thm:informal-main}
For unlimited-supply auctions with a prediction of the optimal uniform price, no universally truthful mechanism can have better revenue guarantees than mechanisms which solely randomize between posting the predicted price and running an optimal prior-free auction.
\end{theorem}

A mechanism is {\em universally truthful} if truthful bidding is always a weakly-dominant strategy for every bidder regardless of randomness and regardless of what other bidders do. This is ensured by offering a price to each bidder independent of their own bid. A {\em prior-free} auction operates without assuming any underlying probability distributions over buyer valuations, instead guaranteeing strong competitive revenue relative to optimal benchmarks across all possible inputs. The randomization factor, which we denote as $\alpha$, determines how much the mechanism trusts the prediction. Placing greater weight on the predicted price improves revenue when the prediction is correct but weakens the guarantee when it is wrong and vice versa for placing greater weight on the prior-free auction. Adjusting this probability  traces the entire consistency-robustness Pareto-optimal frontier.

This theorem is stronger than simply showing that this mixture performs well. The lower bound applies to every universally truthful mechanism, including mechanisms that transform the prediction, correlate prices across bidders, or use the prediction differently across valuation profiles. None of the more convoluted approaches can improve one guarantee without weakening the other. Thus, despite receiving less information than a predicted valuation profile, the optimal use of a price prediction admits a complete and  simple characterization. In contrast to value-prediction results that establish achievability, we obtain a matching lower bound proving no universally truthful bid-independent mechanism does better.

We also complement this structural result in two directions. First, we extend our consistency-robustness tradeoff to handle graceful degradation with price prediction error. Second, we close the statistical loop by studying how the predicted price can be learned from historical markets. The remainder of this section states these contributions precisely and positions them relative to prior work.

\subsection{Our contributions}
\label{sec:contributions}
We now introduce the notation and benchmarks that contextualize our results. For a valuation profile $v$, let $v_{(k)}$ be the $k$th-highest value. The full uniform-price benchmark is
\[
    \F(v)=\max_{1\le k\le n} k\,v_{(k)},
\]
the maximum revenue obtained by posting one price to all bidders. The more relaxed benchmark is
\[
    \Ftwo(v)=\max_{2\le k\le n} k\,v_{(k)},
\]
which requires at least two bidders to buy at the optimal price. \cref{rem:benchmark} shows this relaxed benchmark is necessary since it removes the impossible demand that a prior-free truthful auction extract a constant fraction of one bidder whose value is arbitrarily larger than the others. Let $\lambda_n$ denote the smallest approximation ratio for which a universally truthful bid-independent auction
can guarantee $\frac{1}{\lambda_n} \cdot \Ftwo(v)$ on every valuation profile, as characterized by \citet{ChenGravinLu2014}. In particular, $\lambda_2=2$, and $\lambda_n$ increases monotonically to approximately $2.42$. 

Given a predicted price $\hat p$ and constants $c,r \in [0,1]$, a mechanism is \emph{$c$-consistent} if it earns expected revenue at least $c \cdot \F(v)$ whenever $\hat p$ is an optimal uniform price, and it is \emph{$r$-robust} if it earns at least $r \cdot \Ftwo(v)$ for every valuation profile and every prediction. 

We construct a randomized mechanism that offers $\hat p$ to every bidder with probability $\alpha$ and, with probability $1-\alpha$, ignores the prediction and runs an optimal $\lambda_n$-competitive prior-free auction. This simple mechanism guarantees
\[
    c=\alpha,
    \qquad
    r=\frac{1-\alpha}{\lambda_n}.
\]
Our main result, informally described in \cref{thm:informal-main}, shows that the above mechanism is Pareto-optimal with respect to the consistency–robustness tradeoff for each $\alpha\in [0, 1]$, and varying $\alpha$ traces the entire Pareto frontier. \cref{thm:informal-main} is formally stated as \cref{thm:n-frontier}. We now outline the rest of our contributions below.

\paragraph{Deriving the consistency-robustness Pareto frontier (Sections~\ref{sec:frontier} and~\ref{sec:n}).}
The main ingredient in the proof of \cref{thm:n-frontier} is a novel budgeting lemma stated in \cref{lem:budget}. It upper bounds the normalized expected revenue of any random price distribution offered by the mechanism over a selected range of bidder values. Applying the lemma separately to each of the $n$ bidders shows the mechanism has at most $n$ units of total budget. If the prediction is $H$, robustness on profiles where all values are below $H/n$ consumes roughly $n \cdot \lambda_n \cdot r$ units, while consistency on profiles with one value equal to $H$ consumes roughly $n \cdot c$ units. Thus, we conclude $n\lambda_n r + nc \le n$, and dividing by $n$ gives our final Pareto frontier,
\[
    c+\lambda_n r\le1,
\]
attained by the optimal mechanism outlined above.

\paragraph{Tolerance to imperfect predictions (Section~\ref{sec:error}).} Beyond the consistency-robustness tradeoff that characterizes the extremes of exactly correct and arbitrarily wrong predictions, we also study how mechanism performance degrades relative to intermediate prediction errors in \cref{sec:error}. Let $\rho(s)$ denote the guaranteed fraction of $\F$ with log-error $s$ given a prediction that overestimates the optimal monopoly price. Then, we show that every mechanism satisfies
\[
    \lambda_n r+c+\int_0^\infty \rho(s)\,ds\le1,
\]
and we furthermore show that this condition is tight---we present a {\em randomized discounting} mechanism that achieves the Pareto bound. Thus, robustness, consistency, and error tolerance all stem from the same pricing budget.

\paragraph{Learning a shared (market-agnostic) price (Section~\ref{sec:learning-prices}).}
Having characterized how a price prediction should be used, we conclude by showing how to {\em learn} price predictions from data. In \cref{sec:learning-prices}, we study the class of revenue functions
\[
    \mathcal U_n
    =
    \left\{
        v\mapsto p\cdot\left|\{i:v_i\ge p\}\right|
        :p\ge0
    \right\}
\] induced by posted prices
and show its pseudo-dimension is $\Pdim(\mathcal U_n)=\Theta(\log n)$.
Consequently, when values lie in $[0,H]$,
\cref{cor:shared-price-learning-main} shows that
\[
    O\!\left(
        \frac{\log n+\log(1/\delta)}{\epsilon^2}
    \right)
\]
historical markets suffice to achieve expected revenue within additive error $\epsilon nH$ of the best shared price. Finally, \cref{cor:end-to-end-learning} combines the learned prediction with the best prior-free mechanism, giving a direct end-to-end revenue guarantee.

\paragraph{Contextual (market-specific) learning (Section~\ref{sec:learning-prices}).}
Finally, when public market features are available, we extend the analysis above to a class $\mathcal H$ of market-specific, or {\em contextual}, price predictors and show that
$
    \Pdim(\mathcal U_{\mathcal H})
    =
    O\!\left(d_{\mathcal H}\log n\right),
$
where $d_{\mathcal H}=\Pdim(\mathcal H)$ (\cref{thm:contextual-price-learning}).

{\footnotesize
\begin{table}[t]
\centering
\caption{Consistency-robustness tradeoffs, pseudodimension, and sample complexity bounds for auctions with price predictions.}
\label{tab:results}
\begin{tabular}{@{}lll@{}}
\toprule
Auction setting & Tradeoff & Optimal mechanism \\
\midrule
2 bidders with exact prediction & $c+2r\le1$ & Post prediction + prior-free\\
$n$ bidders with exact prediction & $c+\lambda_n r\le1$ & Post prediction + prior-free \\
Imperfect prediction
& $\lambda_n r+c+\int_0^\infty \rho(s)\,ds\le1$
& Log-discount pricing + prior-free \\
\bottomrule
\end{tabular}

\vspace{0.2in}

\begin{tabular}{@{}lll@{}}
\toprule
Learning setting & Pseudodimension & Sample complexity \\
\midrule
Shared-price learning & $\Pdim(\mathcal U_n)=\Theta(\log n)$ &
$T=O((\log n+\log(1/\delta))/\epsilon^2)$ \\
Contextual revenue learning & $\Pdim(\mathcal U_{\mathcal H})=O(d_{\mathcal H}\log n)$ &
$T=O((d_{\mathcal H}\log n+\log(1/\delta))/\epsilon^2)$ \\
\bottomrule
\end{tabular}

\end{table}
}
\paragraph{}\Cref{tab:results} succinctly summarizes our above contributions. We position our contributions with prior work in Section~\ref{sec:related}. Section~\ref{sec:model} introduces our model and the randomized prediction-fallback mechanism. Sections~\ref{sec:frontier} and~\ref{sec:n} develop the budgeting argument and prove the optimal consistency-robustness frontier, first for two bidders and then for the general case of $n$ bidders.
Section~\ref{sec:error} gives a tight characterization of graceful revenue degradation under imperfect predictions, while Section~\ref{sec:learning-prices} studies how price predictions can be learned from historical data, including sample-complexity guarantees for shared and contextual pricing rules. Section~\ref{sec:conclusion} concludes with directions for future work. \looseness-1
\subsection{Related work}
\label{sec:related}

\paragraph{Competitive analysis of auctions.}
Prior-free mechanism design asks for truthful auctions that perform well on every valuation profile, without assuming that bidder values are drawn from a known distribution. For unlimited-supply goods, this program has led to competitive auctions evaluated against uniform-price revenue benchmarks and to bid-independent posted-price constructions \citep{FiatEtAl2002,GoldbergEtAl2006}. The random sampling optimal price (RSOP) mechanism randomly partitions the bidders into two groups, estimates a profitable uniform price for each group, and applies the price of one group to the other.

Subsequent work sharpened the analysis of this benchmark and the mechanisms competing against it. \citet{AlaeiMalekianSrinivasan2013} proved that RSOP has approximation ratio at most $4.68$ relative to $\Ftwo$. \citet{ChenGravinLu2014} then moved from analyzing a particular auction to identifying the optimal approximation ratio against $\Ftwo$. They characterize the exact ratio $\lambda_n$ and use a product equal-revenue distribution in the corresponding lower bound. A more detailed explanation of their optimal mechanism is outlined in \cref{subsec:chen-mech}.

\paragraph{Learning-augmented mechanism design.} \citet{balcan2005mechanism} introduced machine learning for mechanism design.
There has recently been a rise in popularity surrounding learning-augmented mechanism design where researchers seek to utilize correct information to improve guarantees while still achieving performance relative to prior-free results when the information is incorrect. \citet{LykourisVassilvitskii2018} formalized this tension, outlining the consistency-robustness framework within learning-augmented algorithms. This perspective has since been applied to strategic settings including facility location, scheduling, and clock auctions~\citep{AgrawalEtAl2022,BalkanskiGkatzelisTan2023,GkatzelisSchoepflinTan2025}. These papers share the goal of consistency and robustness, but the predicted object and the economic benchmark depend on the application. More recent work broadens this framework to randomized and robust facility location, budget-feasible mechanisms, mechanisms with outliers, general multi-dimensional mechanism design, and other strategic objectives \citep{BalkanskiGkatzelisShahkarami2024,BarakGuptaTalgamCohen2024,AmanatidisEtAl2025,DeligkasEtAl2025,GoldnerMohanTsilivis2026,BalcanPrasadSandholm2023}.

\paragraph{Auctions with predictions.}

The literature on auction design with predictions generally assumes access to predictions about bidder values. \citet{XuLu2022} initiate a systematic study of mechanism design with predictions and identify tradeoffs among consistency, robustness, and error tolerance. Building on this perspective, \citet{LuWanZhang2024} study competitive auctions supplied with predictions of bidders' private values, including unlimited-supply, limited-supply, and downward-closed environments. \citet{BalkanskiEtAl2024,CaragiannisKalantzis2024} assume predictions of the highest bidder value in online or randomized single-item auctions. These papers predict private values, whereas we predict the revenue-maximizing uniform price.~\citet{BalcanPrasadSandholm2023,BalcanPrasadSandholm2025,PrasadBalcanSandholm2026,prasad2025revenue} develop a general framework for mechanism design with predictions (with applications to combinatorial auctions) that enables highly expressive predictions about bidders, and show how those predictions can be learned from historical agent data. Finally, \citet{ChristodoulouSgouritsaVlachos2024} distinguish predictions of private types, which they call \emph{input advice}, from recommendations about a mechanism's outcome, which they call \emph{output advice}. Our predicted uniform price $\hat{p}$ can be viewed as a form of output advice. To the best of our knowledge, our paper is the first to study auctions under an output price prediction model.

Compared to \citet{LuWanZhang2024} that assumes predictions on bidders' values, our guarantees are weaker, which is expected given that the prediction in our model is a single scalar $\hat{p}$ while theirs assumes predictions on every bidder's value. Given a prediction of the full valuation profile, \citet{LuWanZhang2024} obtain 1-consistency against a stronger benchmark $\OPT(v)=\sum_{i=1}^n v_i$ and robustness factor $\lambda_n+2$ against $\Ftwo$ in the unlimited-supply setting. Thus, full value predictions allow for constant robustness even under perfect consistency. In contrast, with price predictions constant-factor robustness is impossible to achieve alongside 1-consistency.\looseness-1

\paragraph{Learning auction parameters from samples.}
With the aim of providing end to end guarantees, it is also crucial to specify how the predictions are sourced. To characterize the difficulty of obtaining predictions, complementary literature studies how much historical data is needed to learn a revenue-maximizing auction. \citet{ColeRoughgarden2014} study the sample complexity of approximately optimal revenue maximization. \citet{MorgensternRoughgarden2015} develop pseudo-dimension bounds for structured auction classes, shifting the emphasis from learning an unrestricted value distribution to controlling the complexity of the auction family being optimized. Subsequent work gives more general tools for deriving uniform-convergence guarantees from the structure of parameterized algorithm and mechanism classes \citep{balcan_piecewise_decomposable,BalcanSandholmVitercik2023}.\looseness-1

Our learning result is a tight specialization of this program to the elementary family of shared uniform prices. Although the family has only one real-valued parameter, a market with $n$ bidders induces a revenue curve with as many as $n$ demand changes. We show that its pseudo-dimension is $\Theta(\log n)$. While the upper bound uses known techniques from prior work, the lower bound is new in the context of data-driven mechanism design. Tight lower bounds are known in some other applications of data-driven algorithm design~\citep{balcan2017learning,balcan2022provably,balcan2024new,Du2025TuningAA}, but these typically involve problem-specific constructions of instances that are challenging to learn. We then extend the upper-bound argument to contextual price predictors, obtaining pseudo-dimension $O(d_{\mathcal H}\log n)$. In both cases, the learner optimizes the downstream revenue of the price rule directly, and the learned price can be combined with optimal prior-free mechanisms. While we mainly focus on statistical efficiency, some prior works~\citep{sharma2023efficiently,balcan2024accelerating,balcan2024newguarantees} give interesting computational efficiency guarantees for pricing problems, and it would be interesting to extend their techniques to our setting. Another interesting direction would be to learn prices online by extending prior work on online data-driven algorithm design~\citep{balcan2018dispersion,Balcan2021LearningtolearnNP,sharma2020learning}.
\section{Model of price predictions}
\label{sec:model}

We introduce our auction model and desired guarantees. First, we describe the unlimited-supply setting and prior-free revenue benchmarks. Next, we introduce our price-prediction model and use it to define consistency under a correct prediction and robustness under an arbitrary prediction. We end by showing why randomization is necessary in this setting and present the randomized mechanism that drives our main results.

\subsection{Unlimited-supply environment}

There are $n\ge2$ unit-demand bidders and an unlimited supply of identical items.  Bidder $i$ has private value $v_i>0$ for one copy and quasi-linear utility.  We write $v_{(1)}\ge\cdots\ge v_{(n)}$ for the order statistics; bidder labels themselves are not reordered. For a uniform price $p$, define
\[
    N_v(p)\coloneq \bigl|\{i:v_i\ge p\}\bigr|,
    \qquad
    \Rev(v;p)\coloneq p \cdot N_v(p).
\]
It is enough to optimize the posted price over reported values, since the number of buyers is constant between consecutive values.

\begin{definition}[Uniform-price benchmarks]
\label{def:benchmarks}
The full monopoly benchmark and its standard prior-free truncation are
\begin{align}
    \F(v)&\coloneq \max_{1\le k\le n}k v_{(k)},
    \label{eq:F}\tag{F}\\
    \Ftwo(v)&\coloneq \max_{2\le k\le n}k v_{(k)}.
    \label{eq:F2}\tag{$F^{(2)}$}
\end{align}
An \emph{optimal price} is any $p^*(v)$ satisfying
$\Rev(v;p^*(v))=\F(v)$.
\end{definition}

The full benchmark $\F (v)$ allows a sale to only the highest bidder, but the relaxed benchmark $\Ftwo$ enforces at least two bidders to buy at the uniform price. This small modification is essential since no truthful auction can extract a constant fraction of an arbitrarily isolated highest value without any prior information. By contrast, $\Ftwo$ admits a constant approximation ratio and is the canonical benchmark in the unlimited-supply auction literature \citep{GoldbergEtAl2006,ChenGravinLu2014}.

\begin{remark}[Stronger benchmarks] \label{rem:benchmark}
The full benchmark $\F$ is stronger than $\Ftwo$ but does not have a finite prior-free approximation ratio without bounds on the values. Thus, requiring robustness against $\F$ makes the problem infeasible.  In general symmetric feasibility environments, $\mathrm{EFO}^{(2)}$ is the accepted extension of this truncation
\citep{HartlineYan2011}, and in unlimited supply, $\Ftwo$ is the appropriate specialization.

There is a stronger monotone-price benchmark, usually denoted $\mathcal{M}^{(2)}$, for \emph{ordered-bidder} markets \citep{BhattacharyaEtAl2013}. It assumes that bidders have some natural, meaningful rank (like high-value to low-value) and compares an auction's performance against personalized, descending prices set for each bidder. But, without such an order, the symmetric unlimited-supply model studied here has $\Ftwo$ as its canonical benchmark.\looseness-1
\end{remark}

\subsection{Optimal prior-free approximation ratio for truthful auctions} \label{sec:truthful}

We now describe the characterization of~\citet{ChenGravinLu2014} of the optimal performance for universally truthful mechanisms without priors. A mechanism is universally truthful if each bidder, conditional on the other bids and the mechanism's internal random seed, is offered a price independent of their own bid. The bidder receives one item and pays the price offered by the mechanism if and only if their bid is at least the price. Such a mechanism is a distribution over deterministic dominant-strategy incentive-compatible and individually rational mechanisms \citet{GoldbergEtAl2006,ChenGravinLu2014}. We impose no positive transfers and normalize the payment of a bidder who does not buy to zero.\looseness-1

Let $\lambda_n$ denote the smallest approximation ratio for which a universally truthful bid-independent auction
can guarantee $\frac{1}{\lambda_n} \cdot \Ftwo(v)$ on every valuation profile. \citet{ChenGravinLu2014} show that

\begin{equation}
    \lambda_n
    =1-\sum_{k=2}^{n}
       \left(-\frac1n\right)^{k-1}
       \frac{k}{k-1}\binom{n-1}{k-1}.
    \label{eq:lambda}
\end{equation}
In particular,
    $\lambda_2=2,
    \lambda_3=\frac{13}{6},
    \lambda_4=\frac{215}{96},$
and $\lambda_n$ increases to approximately $2.42$. In this paper, let $\A_n^*$ denote an optimal prior-free mechanism that achieves approximation ratio $\lambda_n$.

\citet{ChenGravinLu2014} also identifies the product equal-revenue distribution as the worst-case distribution. For later use, its exact identity is
\begin{equation}
    \int_{[1,\infty)^n}
        \Ftwo(v)\prod_{i=1}^n\frac{dv_i}{v_i^2}
    =n\lambda_n.
    \label{eq:equal-revenue-identity}
\end{equation}
In words, under this distribution, every posted price yields the exact same expected revenue since a price $p$ is accepted with probability $\frac1 p$. Since every price yields identical expected revenue, this distribution removes any clear statistical structure to exploit and therefore is the hardest possible instance for an auction to perform well against.

\subsection{Price predictions} \label{sec:price-predictions}

We now describe our main model of price predictions. Before receiving the current bids, the seller receives a price recommendation $\hat p\in (0,\infty)$.  The recommendation may be constructed from historical data and public covariates, but it must not depend on any current bidder's report (so as to preserve truthfulness). The prediction $\hat p$ is \emph{correct} on profile $v$ if $\hat p$ is an optimal price for $v$, noting that ties among optimal prices are allowed.

\paragraph{Price predictions are easier to learn than value predictions.} We claimed in \cref{sec:introduction} that price predictions are more realistic to obtain than bidder-wise value predictions. Here we construct a simple problem instance in which learning the optimal monopoly price is trivial while learning bidders' values is impossible. A more in-depth investigation of how to obtain learned prices is in \cref{sec:learning-prices} and \cref{sec:learning}. 

\begin{example}
Fix $n\ge2$ and choose any $h\in(1,n/(n-1))$. In every market, bidder $1$ has value $1$, while each bidder $i\ge2$ independently has value either $1$ or $h$, each with probability $1/2$. Note that $p^*=1$ is an optimal uniform price since it earns $n$ and every price $p>1$ sells to at most $n-1$ bidders, so it only earns $(n-1)h<n$ for every realization. Without observing the current bidders, the optimal price is perfectly predictable, but since the realized values of bidders $2,\ldots,n$ are independent fair draws, no predictor based only on past markets can exactly determine those current bidder values.
\end{example}

\subsection{Consistency, robustness, and failure of deterministic mechanisms}
Before we present our main mechanism, we first show why a pricing rule must perform randomization. We begin with formal definitions for consistency and robustness.

\begin{definition}[Consistency and robustness]
\label{def:guarantees}
A mechanism $\M$ has \emph{consistency} $c\in[0,1]$ if
\[
    \E[\Rev(\M(v,\hat p))]\ge c \cdot \F(v)
    \qquad
    \text{whenever }\Rev(v;\hat p)=\F(v).
\]
It has \emph{robustness} $r\in[0,1]$ if
\[
    \E[\Rev(\M(v,\hat p))]\ge r \cdot \Ftwo(v)
    \qquad
    \text{for every }v\text{ and every }\hat p.
\]

\end{definition}

The two guarantees deliberately use different benchmarks. Correct advice identifies a price earning the full benchmark $\F$, which is impossible under arbitrary advice, requiring the prior-free benchmark $\Ftwo$.

The following theorem shows truthful deterministic mechanisms fail to achieve any constant robustness and consistency. 

\begin{theorem}[Deterministic mechanisms fail]
\label{thm:deterministic-fail}
For every $n \ge 2$, no deterministic bid-independent mechanism can simultaneously have consistency $c>0$ with respect to $\F$ and robustness $r>0$ with respect to $\Ftwo$.
\end{theorem}

\begin{proof}
Assume for contradiction that such a mechanism exists.

Define $q_i(v_{-i};\hat p)$ as the price given to bidder $i$, which is a function of all of the other bids and the public prediction $\hat p$. Then, fix $L>0$, choose $H>\max\left\{nL,\frac{nL}{c}\right\}$, and fix the public prediction $\hat p=H$. 

For a bidder $i$, consider the valuation profile where $v_i=H$ and $v_j=L$ for every $j \ne i$. Since $H>L$, posting $H$ earns $H$ in total from bidder $i$, whereas posting any other price accepted by all $n$ bidders can earn at most $nL<H$. Any $p\le L$ earns at most \(nL<H\), any \(L<p<H\) earns \(p<H\), and any \(p>H\) earns zero. Thus, the prediction $\hat p = H$ is correct, and consistency requires revenue at least $cH$.
Now, since the $n-1$ bidders other than $i$ can contribute at most $(n-1)L$, the price offered to bidder $i$ must satisfy\looseness-1

\[ 
    q_i(L,\ldots,L;H) \ge cH-(n-1)L > L
\]
to be $c-$consistent.

Then, consider a different valuation profile where every bidder has value $L$, and let the prediction $\hat p$ remain at $\hat p = H$. Now, the prediction is incorrect since offering $L$ would be the optimal price that earns revenue $nL$. The condition of $r$-robustness demands revenue at least $rnL$, but each bidder $i$ is offered the same $q_i(L,\ldots,L;H)>L$ as before, so every bidder rejects. The mechanism therefore earns zero revenue contradicting $r$-robustness for every $r>0$.
\end{proof}
We conclude that randomization is necessary to jointly achieve constant robustness and consistency.

\subsection{The randomizing mechanism}
\label{sec:upper}

We first introduce a shorthand to help characterize our results. For a distribution $Q$ over nonnegative prices and a bidder with value $x$, define the expected revenue from the bidder when offered a price from $Q$ as,
\begin{equation}
    \rev_x(Q)\coloneq \E_{P\sim Q}\bigl[P\cdot \ind\{P\le x\}\bigr].
    \label{eq:truncated-revenue}
\end{equation}
This is the only auxiliary revenue notation used in the lower bound.
In this section we define a family of randomized mechanisms that we later show are Pareto-optimal for consistency and robustness. Let $\A$ be any universally truthful prior-free mechanism that is
$\lambda$-competitive against $\Ftwo$, that is,\looseness-1
\[
    \E[\Rev(\A(v))]\ge \frac{\Ftwo(v)}{\lambda}
    \qquad\text{for every }v.
\]

\begin{definition}
\label{def:mixture}
For $\alpha\in[0,1]$, the mechanism $\M_{\alpha,\A}$ does the following:
\begin{enumerate}[leftmargin=2em]
    \item with probability $\alpha$, offer $\hat p$ to every bidder;
    \item with probability $1-\alpha$, run $\A$.
\end{enumerate}
\end{definition}

\noindent In the context of auctions with price predictions, we refer to a prior-free mechanism $\mathcal A$ that is robust to inaccurate predictions as the \emph{fallback} mechanism. We now analyze the consistency and robustness of the mechanism $M_{\alpha, \A}$.

\begin{proposition}
\label{prop:generic-upper}
The mechanism $\M_{\alpha,\A}$ is universally truthful, has consistency at least $\alpha$ with respect to $\F$,
and has robustness at least $\frac{1-\alpha}{\lambda}$ with respect to $\Ftwo$.
\end{proposition}

\begin{proof}
First we show that $\M_{\alpha,\A}$ is truthful. The random branch factor $\alpha$ is chosen independently of all bids. Since both branches are universally truthful, randomizing between them is universally truthful.

Next we show the consistency and robustness guarantees. If the prediction $\hat p$ is correct, the prediction branch earns exactly
$\Rev(v;\hat p)=\F(v)$, so expected revenue is at least
$\alpha \cdot \F(v)$. When the prediction is incorrect, we ignore the prediction branch and bound the revenue of the fallback $\mathcal A$:
\[
    \E[\Rev(\M_{\alpha,\A}(v,\hat p))]
    \ge (1-\alpha) \cdot \frac{\Ftwo(v)}{\lambda}.
\]
We conclude that $\M_{\alpha,\A}$ is $\alpha$-consistent and $\frac{1-\alpha}{\lambda}$-robust.
\end{proof}

Using the optimal truthful prior-free mechanism $\A = \A_n^*$ with $\lambda=\lambda_n$ yields the family of randomized
mechanisms $\M_{\alpha, \A_n^*}$ that we will use in our main result.

\paragraph{Cross-pricing fallback for 2 bidders.} \label{sec:cross-pricing}

For two bidders, a particularly simple optimal fallback is \emph{cross-pricing}: offer bidder $1$
bidder $2$'s bid and offer bidder $2$ bidder $1$'s bid.  If
$v_{(1)}>v_{(2)}$, only the highest bidder accepts, producing revenue $v_{(2)}$. If the values tie,
both accept.  In every case the revenue is at least
\[
    v_{(2)}=\frac12 \cdot \Ftwo(v),
\]
so cross-pricing achieves the optimal approximation ratio $\lambda_2 = 2$.

\paragraph{What is the optimal \texorpdfstring{$n$}{n}-bidder fallback?} \label{subsec:chen-mech}

The optimal truthful prior-free mechanism $\A_n^*$ derived in \citet{ChenGravinLu2014} is most naturally understood through
its price-probability feasibility system.  On a finite multiplicative value grid, let
$z_i(b_{-i},p)$ be the probability of offering price $p$ to bidder $i$ when the other bids are
$b_{-i}$.  A feasible prior-free pricing rule satisfies
\begin{align}
    &z_i(b_{-i},p)\ge0,
    \qquad
    \sum_p z_i(b_{-i},p)\le1,
    \label{eq:CGL-probability}\\
    &\sum_{i=1}^n\sum_{p\le b_i}p\,z_i(b_{-i},p)
      \ge \frac{\Ftwo(b)}{\lambda}
      \qquad\text{for every bid vector }b.
    \label{eq:CGL-revenue}
\end{align}
\citet{ChenGravinLu2014} characterize when this system is feasible and prove feasibility at
$\lambda=\lambda_n$.

Sampling from the resulting bidder-specific price distributions implements a universally truthful
auction.

This characterization is exact but is not a short closed-form auction for general $n$.  On a
discretized market it can be implemented by solving the associated feasibility problem, but the
number of bid profiles and conditional price distributions grows rapidly with the grid and the
number of bidders.  It is therefore appropriate as an optimal theoretical black box, while a seller
may prefer a simpler fallback in practice.

\paragraph{Simple practical fallbacks.}
The random sampling optimal price (RSOP) mechanism defined in \citet{GoldbergEtAl2006} is easy to implement: randomly split bidders into two groups, compute the best
uniform price in each group, and post that price to the other group. \citet{AlaeiMalekianSrinivasan2013} showed that RSOP's approximation ratio is at most
$4.68$ against $\Ftwo$, but it is widely conjectured that in fact RSOP achieves approximation ratio 4.

A simple mechanism with a better approximation ratio is the sampling cost sharing (SCS) auction of
\citet{FiatEtAl2002}.  Randomly partition the bidders into $A$ and $B$, compute
$\F(A)$ and $\F(B)$, and use each group to set a target revenue for the other.  Given a target $T$
and a group of bids, the cost-sharing subroutine chooses the largest $k$ for which the $k$ highest
bids are at least $\frac{T}{k}$, and charges each of those bidders $\frac{T}{k}$. SCS is truthful and
has approximation ratio 4 against $\Ftwo$.  Using SCS in \cref{def:mixture} gives
\[
    c=\alpha,\qquad r=\frac{1-\alpha}{4}.
\]
This is weaker than the optimal $r=\frac{1-\alpha}{\lambda_n}$ but is computationally
feasible.
\section{Deriving the consistency-robustness Pareto frontier} \label{sec:frontier}

In \cref{sec:upper} we showed that the randomized mechanism $\M_{\alpha,\A^*_n}$ is $\alpha$-consistent and $\frac{1-\alpha}{\lambda_n}$ robust. In this section we prove that this family of mechanisms for $\alpha\in [0, 1]$ is Pareto-optimal with respect to consistency and robustness. The main difficulty is demonstrating the lower bound rules out every other bid-independent way of using the prediction. Our mechanism offers either a point mass on $\hat p$ or an optimal prior-free mechanism, but others could utilize any valid probability distribution that depends on $\hat p$ or a more involved combination of $\hat p$ and prior-free pricing rules.

\subsection{A budgeting lemma}
\label{sec:budget}

The key technical ingredient in the proof of \cref{thm:informal-main} is a novel budgeting lemma, which converts the requirement of robust revenue across many valuation profiles into a constraint on the price distribution offered to the bidder. For intuition, when a bidder receives a random price from a distribution $Q$, independent of their own value, robustness requires the distribution to earn enough revenue for many possible values of that bidder. On the other hand, consistency requires the distribution to place enough mass near the correct price, so the lemma averages $\rev_z(Q)$ over a family of values $z$ and shows that these requirements must share a single unit of normalized revenue.

A similar idea appears in \citet{CaragiannisKalantzis2024}, where they use Myerson's payment identity to show that requiring strong revenue guarantees across many different prediction errors consumes a limited allocation capacity, leading to a feasibility constraint characterized using an integral. While similar, our lemma applies to randomized price distributions and bounds normalized revenue across bidder values. When applied separately to every bidder, it yields the general Pareto frontier.

\begin{lemma}[Budgeting lemma]
\label{lem:budget}
Let $Q$ be any distribution over prices in $\Rplus$, and let $1<L\le H$.  Then
\begin{equation}
    \int_1^L \rev_z(Q) \cdot \frac{dz}{z^2}
    +\frac1H\rev_H(Q)
    \le1.
    \label{eq:budget}
\end{equation}
\end{lemma}

\begin{proof}
Let $P\sim Q$.  Tonelli's theorem applies because the integrand is nonnegative.  Therefore the
left-hand side of \eqref{eq:budget} equals
\[
    \E\left[
        P\left(
            \int_{\max\{1,P\}}^L 1 \cdot \frac{dz}{z^2}
            +\frac{\ind\{P\le H\}}{H}
        \right)
    \right],
\]
where the integral is zero if its lower endpoint exceeds $L$.  It is enough to show that the
quantity inside the expectation is at most one for every realized $P$.
\begin{itemize}[leftmargin=2em]
    \item If $0\le P\le1$, it equals
    $P(1-1/L+1/H)\le P\le1$.
    \item If $1<P\le L$, it equals
    $P(1/P-1/L+1/H)=1-P(1/L-1/H)\le1$.
    \item If $L<P\le H$, the integral vanishes and the quantity is $P/H\le1$.
    \item If $P>H$, both terms vanish.
\end{itemize}
Taking expectations completes the proof.
\end{proof}

The density $z^{-2}$ is chosen so that its tail scales as $1/z$.  Multiplying this tail by a posted
price cancels the scale of the price.  Consequently, every realized price consumes at most one
unit of normalized revenue budget.  A random pricing rule is a convex combination of realized
prices, so it has the same unit budget.  The last term $\rev_H(Q)/H$ reserves part of this
budget for the high value at which the prediction will be correct.  Earning more at the high value leaves less
capacity for the continuum of lower-value profiles needed for robustness.

\begin{remark}[Relation to equal-revenue lower bounds]
The weight $z^{-2}$ is the density of the equal-revenue distribution, for which every deterministic
posted price has the same expected revenue.  Product equal-revenue measures are classical in
prior-free auction lower bounds and are central to the characterization of
\citet{ChenGravinLu2014}. \citet{CaragiannisKalantzis2024} use the same inverse-square integral in an
inverse form of Myerson's payment identity for single-item learning-augmented auctions.  \cref{lem:budget} is a truncated posted-price version
tailored to optimal monopoly price predictions.
\end{remark}

As a quick consequence, robustness against the full benchmark is impossible.

\begin{proposition}[No positive robustness against $\F$]
\label{prop:no-full-robustness}
On an unbounded value domain, no universally truthful bid-independent mechanism has a positive
robustness with respect to the full benchmark $\F$, even when the public advice is fixed.
\end{proposition}

\begin{proof}
Fix the advice and fix the values of bidders $2,\ldots,n$ at $1$.  Let $Q$ be the price distribution
offered to bidder $1$; bid independence makes $Q$ independent of $v_1$.  On the profile
$(h,1,\ldots,1)$, the other bidders contribute at most $n-1$.  If the mechanism guaranteed a
fraction $a>0$ of $\F$, then for all sufficiently large $h$,
\[
    \rev_h(Q)\ge ah-(n-1)\ge \frac a2h.
\]
Apply \cref{lem:budget} with $L=H$ and omit the nonnegative last term. Choose $h_0$ sufficiently large such that $\rev_h(Q)\ge \frac{a}{2}h$ for all $h\ge h_0$. Then,
\[
    1\ge\int_1^H\rev_z(Q) \cdot \frac{dz}{z^2}
      \ge \frac a2\int_{h_0}^H 1 \cdot \frac{dz}{z},
\]
which diverges as $H\to\infty$, a contradiction.
\end{proof}

\subsection{Two-bidder consistency-robustness Pareto frontier}
\label{sec:two}

We first prove the main lower bound for two bidders and then extend to $n$ bidders in \cref{sec:n}. We present the 2-bidder case first because the proof highlights the economic role of the budgeting lemma without the additional notation needed for the product measure in general dimension.

\begin{theorem}[2-bidder Pareto frontier]
\label{thm:two-frontier}
Let $n=2$.  If a universally truthful bid-independent mechanism is $c$-consistent with respect to
$\F$ and $r$-robust with respect to $\Ftwo$, then
\begin{equation}
    c+2r\le1.
    \label{eq:two-frontier}
\end{equation}
\end{theorem}

\begin{proof}
Fix $H>2$ and set the public prediction to $\hat p=H$.  Define a finite measure on values by
\begin{equation}
    \mu_H(dz)
     \coloneq \ind_{[1,H/2]}(z) \cdot \frac{dz}{z^2}
      +\frac1H\delta_H(dz).
    \label{eq:two-measure}
\end{equation}
Its continuous mass and total mass are
\[
    \ell_H \coloneq 1-\frac2H,
    \qquad
    m_H \coloneq \ell_H+\frac1H=1-\frac1H.
\]

\noindent For bidder $i$, fix the opponent's value and let $Q_i$ be the resulting random price offered to
$i$.  \cref{lem:budget}, with $L=H/2$, states that bidder $i$'s revenue integrated over its
own value under $\mu_H$ is at most one.  Integrating the opponent's value contributes a factor
$m_H$.  Summing over the two bidders gives the upper bound
\begin{equation}
    \int \E[\Rev(\M(v,H))]\,\mu_H^{\otimes2}(dv)
    \le2m_H.
    \label{eq:two-upper}
\end{equation}

\noindent We next lower-bound the same integral on two disjoint classes of profiles. First, if both values lie in $[1,H/2]$, robustness gives
\[
    \E[\Rev(\M(v,H))]\ge r \cdot \Ftwo(v)
       =2r\min\{v_1,v_2\}.
\]
The contribution from this square is at least $rI_{2,H}$, where
\begin{align}
    I_{2,H}
    & \coloneq \int_1^{H/2}\!\int_1^{H/2}
       2\min\{u,s\}\frac{du\,ds}{u^2s^2} \\
    &=4\left(1-\frac2H\right)
      -\frac8H\log\left(\frac H2\right).
    \label{eq:I2H}
\end{align}
The second line follows by symmetry: on the half-square $s\le u$,
$\min\{u,s\}=s$.

\vspace{\baselineskip}

\noindent Second, suppose exactly one value equals $H$ and the other lies in $[1,H/2]$.  Posting $H$
earns $H$, and no uniform price can earn more than $H$ because $2s\le H$.  Thus $H$ is an
optimal price, the advice is correct, and consistency gives revenue at least $cH$.  For each choice
of the high bidder, the atom has mass $1/H$, so this class contributes at least
$c\ell_H$.  There are two choices of the high bidder, for a total consistency contribution
$2c\ell_H$.

All omitted profile classes have nonnegative revenue.  Combining the lower bounds with
\eqref{eq:two-upper} yields
\[
    rI_{2,H}+2c\ell_H\le2m_H.
\]
Let $H\to\infty$.  Then $I_{2,H}\to4$ and $\ell_H,m_H\to1$, so $4r+2c\le2$ as desired.
\end{proof}

\noindent The lower bound is tight. Using cross-pricing as the fallback in our randomized mechanism attains every point on this frontier.

\begin{corollary}[2-bidder achievability]
\label{cor:two-achieve}
For every $\alpha\in[0,1]$, randomizing between the predicted price and cross-pricing yields
\[
    c=\alpha,
    \qquad
    r=\frac{1-\alpha}{2},
\]
and hence achieves equality in \eqref{eq:two-frontier}. 

\end{corollary}

\section{The multi-bidder consistency-robustness Pareto frontier}
\label{sec:n}

The 2-bidder proof extends cleanly by writing the hard measure as a product.  The low-value
cutoff becomes $H/n$, ensuring that a single bidder of value $H$ makes $H$ an optimal uniform
price even if all other bidders are as large as the cutoff.

\begin{lemma}[$n$-bidder hard measure]
\label{lem:n-measure}
Fix $H>n$ and define
\begin{equation}
    \mu_{n,H}(dz)
     \coloneq \ind_{[1,H/n]}(z) \cdot \frac{dz}{z^2}
      +\frac1H\delta_H(dz).
    \label{eq:n-measure}
\end{equation}
Let
\[
    \ell_{n,H} \coloneq 1-\frac nH,
    \qquad
    m_{n,H} \coloneq 1-\frac{n-1}{H}
\]
be its continuous and total masses.  For every universally truthful bid-independent mechanism and
fixed prediction $H$,
\begin{equation}
    \int \E[\Rev(\M(v,H))]\,\mu_{n,H}^{\otimes n}(dv)
    \le n m_{n,H}^{n-1}.
    \label{eq:n-upper}
\end{equation}
\end{lemma}

\begin{proof}
Fix bidder $i$ and the other values $v_{-i}$.  Bid independence gives a price distribution
$Q_i(v_{-i},H)$ that does not depend on $v_i$.  Applying \cref{lem:budget} with
$L=H/n$ gives
\[
    \int \rev_{v_i}(Q_i(v_{-i},H))\,\mu_{n,H}(dv_i)\le1.
\]
Integrating over $v_{-i}$ multiplies the right-hand side by the total mass
$m_{n,H}^{n-1}$.  Sum this inequality over $i=1,\ldots,n$.  Expected total revenue is the sum of
the bidders' expected payments, proving \eqref{eq:n-upper}.
\end{proof}

\noindent Using the product-measure budget from Lemma~\ref{lem:n-measure}, we can now extend the two-bidder lower bound to general $n$ and obtain the full consistency-robustness frontier.

\begin{theorem}[$n$-bidder Pareto frontier]
\label{thm:n-frontier}
For every $n\ge2$, if a universally truthful bid-independent mechanism is $c$-consistent with
respect to $\F$ and $r$-robust with respect to $\Ftwo$, then
\begin{equation}
    c+\lambda_n r \le 1.
    \label{eq:n-frontier}
\end{equation}

\end{theorem}

\begin{proof}
Fix $H>n$, set $\hat p=H$, and use the product measure from
\cref{lem:n-measure}.  Define
\begin{equation}
    I_{n,H} \coloneq 
    \int_{[1,H/n]^n}
        \Ftwo(v)\prod_{i=1}^n\frac{dv_i}{v_i^2}.
    \label{eq:InH}
\end{equation}

\noindent On the low cube $[1,H/n]^n$, robustness contributes at least $rI_{n,H}$ to the revenue integral.  Next consider profiles with exactly one coordinate equal to $H$ and every other
coordinate in $[1,H/n]$.  For every $k\ge2$,
\[
    k v_{(k)}\le n\frac Hn=H.
\]
Hence posting $H$ earns $H=\F(v)$ and is an optimal price.  Consistency gives revenue at least
$cH$.  For a fixed choice of the high bidder, the atom at $H$ has mass $1/H$ and the other
$n-1$ coordinates have total low mass $\ell_{n,H}^{n-1}$, producing contribution
$c\ell_{n,H}^{n-1}$.  There are $n$ choices for the high bidder.  Therefore,
\begin{equation}
    \int \E[\Rev(\M(v,H))]\,\mu_{n,H}^{\otimes n}(dv)
    \ge rI_{n,H}+nc\ell_{n,H}^{n-1}.
    \label{eq:n-lower}
\end{equation}

\noindent Combining \eqref{eq:n-upper} and \eqref{eq:n-lower},
\begin{equation}
    rI_{n,H}+nc\ell_{n,H}^{n-1}
    \le n m_{n,H}^{n-1}.
    \label{eq:n-finite-H}
\end{equation}
As $H\to\infty$, both $\ell_{n,H}$ and $m_{n,H}$ converge to one.  The low cubes increase to
$[1,\infty)^n$, so monotone convergence and the equal-revenue identity
\eqref{eq:equal-revenue-identity} yield
\[
    I_{n,H}\longrightarrow n\lambda_n.
\]
Taking limits in \eqref{eq:n-finite-H} yields $rn\lambda_n+nc\le n$ as desired.
\end{proof}

\noindent Combining this lower bound with Proposition~\ref{prop:generic-upper} gives a complete characterization. The simple prediction-fallback mixture traces the entire Pareto frontier.

\begin{corollary}[Achieving the Pareto frontier] \label{cor:n-achieve}
    For every
    $\alpha\in[0,1]$, the randomized mechanism $\M_{\alpha,\A_n^*}$ in \cref{prop:generic-upper} satisfies
    \[
        c=\alpha,
        \qquad
        r=\frac{1-\alpha}{\lambda_n},
    \]
    and is therefore Pareto optimal (it achieves equality in \eqref{eq:n-frontier}).
\end{corollary}
\section{Error-tolerance in price predictions}
\label{sec:error}

So far, our consistency-robustness results treat two extreme regimes---the prediction is either exactly correct or arbitrarily wrong. A natural question is what revenue can be guaranteed at intermediate error levels, when the prediction $\hat p$ is imperfect but not far from a true optimal price $p^*$. For an optimal price $p^*$, define the \emph{multiplicative prediction error}
\[
\eta(\hat p,p^*) \coloneq \max\left\{\frac{\hat p}{p^*},\frac{p^*}{\hat p}\right\}\ge1.
\]
The mechanism receives $\hat p$, but knows neither $p^*$ nor $\eta$ when it chooses its pricing rule.

It is important to note that underestimation and overestimation have very different consequences. Say every bidder has value $p^*$. If $\hat p=p^*/\eta$, every bidder accepts the underestimate, so offering the prediction earns at least $\F(v)/\eta$. But offering even the smallest overestimate earns zero, since everyone rejects. Any mechanism that degrades gracefully must therefore post prices strictly below the prediction with positive probability. The auction design must reckon with the question of how much probability to spend and at which depths below the prediction.

The budgeting lemma of \cref{sec:budget} characterizes the cost of the aforementioned probability. Each bidder carries one unit of normalized pricing budget, and the Pareto frontier \eqref{eq:n-frontier} spends it on (1) consistency, which reserves mass at the predicted price and (2) robustness, which spreads mass across the continuum of lower values. Error tolerance is the demand that the mechanism also perform at every scale in between, and it draws from the same unit budget.

We measure error on a log-scale for mathematical convenience. In order to guarantee a positive fraction of $F$ on every profile from a prediction that overestimates $p^*$ by a factor $e^s$, the auction must post a price a factor $e^s$ below the prediction, and under the equal-revenue weighting every multiplicative discount is equally expensive. Thus, a unit of probability placed at a deep discount uses the same probability budget, while earning less compared to a shallower discount when the realized error is smaller. Guaranteeing a fraction $a$ of $\F$ against every overestimate up to a factor $R$ therefore costs roughly $a$ at each of the $\log R$ scales in between, or $a\log R$ in total. Adding up all scales, the price of an entire guarantee curve is its area in logarithmic error. The precise accounting of the budget split between consistency, robustness, and (log-)error tolerance is
\[
c+\lambda_n r+\int_0^\infty\rho(s)\,ds\le1,
\]
where $\rho(s)$ is the fraction of $\F$ guaranteed when the prediction is too high by a factor $e^s$. Consistency is the boundary case $s=0$ of the same curve $\rho$, so the inequality says that the area under the error-guarantee curve, together with the robustness term, cannot exceed one unit of budget.

Underestimation requires no additional budget for our matching construction. When $\hat p=p^*/\eta$, every discounted price used to obtain consistency is at most $p^*$, and is therefore accepted by every bidder who would accept $p^*$. Consequently, the construction below automatically obtains a $c/\eta$ guarantee under underestimation. The nontrivial budget tradeoff is therefore on the overestimation side, which is what we characterize below.

We prove the inequality above for every universally truthful bid-independent mechanism in \cref{sec:error-lower-bound}, and then show in \cref{sec:error-construction} that every nonincreasing guarantee curve satisfying it is attained by an explicit mechanism that posts randomly discounted predictions. \cref{sec:error-examples} applies these guarantees to (1) derive same-rate error degradation in both directions (that is, for overestimates and underestimates), (2) obtain tunable error guarantees, and (3) characterize the slowest possible error decay.

\subsection{A necessary condition on error guarantees}
\label{sec:error-lower-bound}

Write $s \coloneq \log(\hat p/p^*)>0$ for the logarithmic error of a strict overestimate, and let
\[
\mathcal P^*(v) \coloneq \arg\max_{p\ge0}\Rev(v;p)
\]
denote the set of optimal uniform prices.

\begin{definition}[Revenue guarantee under overestimation]
\label{def:intrinsic-profile}
A mechanism $M$ \emph{guarantees} a measurable function $\rho:(0,\infty)\to[0,\infty)$ if, for every valuation profile $v$, every $p^*\in\mathcal P^*(v)$, and every $s>0$,
\begin{equation}
\E[\Rev(M(v,e^sp^*))]\ge\rho(s)\F(v).
\label{eq:overestimate-guarantee}
\end{equation}
The expectation is over the mechanism's internal randomization.
\end{definition}

Thus, $\rho(s)$ is the fraction of optimal uniform-price revenue the mechanism promises when the prediction overestimates an optimal price by a factor $e^s$. We emphasize that $\rho(s)$ represents a promise rather than an exact worst case. Any pointwise lower bound $\rho$ is admissible, but larger choices of $\rho$ yield correspondingly stronger implications from Theorem~\ref{thm:unknown-error-lb}. The definitions for consistency $c$ and robustness $r$ are the same as in \cref{def:guarantees}, with consistency playing the role of the boundary case $s=0$.

\begin{theorem}[Consistency, robustness, and prediction error]
\label{thm:unknown-error-lb}
Suppose a universally truthful bid-independent mechanism is $c$-consistent, $r$-robust, and guarantees $\rho$ under overestimation. Then
\begin{equation}
c+\lambda_n r+\int_0^\infty\rho(s)\,ds\le1.
\label{eq:functional-frontier}
\end{equation}
\end{theorem}

Setting $\rho\equiv0$ recovers the consistency-robustness frontier \eqref{eq:n-frontier}, so this theorem strictly strengthens our previous results. The integral is precisely the additional budget that graceful degradation costs, and a mechanism can afford a positive guarantee at intermediate errors only by giving up consistency, robustness, or both.

The integral measures the size of the guarantee and the range of errors over which it is demanded at the same time. In terms of multiplicative error,
\[
\int_0^\infty\rho(s)\,ds=\int_1^\infty\rho(\log\eta)\frac{d\eta}{\eta},
\]
so maintaining a fraction $a$ of $\F$ for every overestimate factor $1<\eta\le R$ requires
\[
c+\lambda_n r+a\log R\le1.
\]
As a result, error tolerance that is either stronger or spread over a wider multiplicative range leaves less room for consistency and robustness guarantees.

The proof extends the budgeting argument of \cref{sec:budget} by integrating revenue over three disjoint classes of profiles: profiles where all values are small which charge the mechanism for robustness, profiles with one larger value and an overestimated prediction which charge for $\rho$, and profiles where the prediction is correct which charge for consistency. The proof below demonstrates that a single upper bound on integrated revenue must accommodate all three.

\begin{proof}
Fix $K>1$, set $H \coloneq nK^2$, and fix the public prediction to $\hat p=H$. Define the finite measure
\begin{equation}
\mu_K(dz) \coloneq \ind_{[1,K]}(z)\frac{dz}{z^2}+\ind_{[nK,H)}(z)\frac{dz}{z^2}+\frac1H\delta_H(dz).
\label{eq:three-band-measure}
\end{equation}
Its total mass is
\[
m_K=1-\left(1-\frac1n\right)\frac1K.
\]
For a deterministic posted price $P$, let $B_K(P) \coloneq P\mu_K([P,\infty))$. Direct calculation gives
\[
B_K(P)=
\begin{cases}
P\left(1-\left(1-\frac1n\right)\frac1K\right)\le1,&0<P\le1,\\[1mm]
1-\left(1-\frac1n\right)\frac{P}{K}\le1,&1<P\le K,\\[1mm]
P/(nK)\le1,&K<P<nK,\\[1mm]
1,&nK\le P\le H,\\[1mm]
0,&P>H.
\end{cases}
\]
Conditional on the other bids $v_{-i}$, bid independence gives bidder $i$ a random price distribution independent of $v_i$. Averaging the preceding bound over that random price, integrating over $v_{-i}$, and summing over bidders yields\looseness-1
\begin{equation}
\int\E[\Rev(M(v,H))]\,\mu_K^{\otimes n}(dv)\le nm_K^{n-1}.
\label{eq:three-band-upper}
\end{equation}

We lower-bound the same integral on three disjoint profile classes. Let $\ell_K \coloneq 1-1/K$.

\paragraph{Robustness.}
On $[1,K]^n$, robustness contributes at least
\[
r\int_{[1,K]^n}\Ftwo(v)\prod_{i=1}^n\frac{dv_i}{v_i^2}.
\]

\paragraph{Overestimation.}
Suppose one coordinate equals $z\in[nK,H)$ and the remaining coordinates lie in $[1,K]$. Posting $z$ earns $z$. Every price serving at least two bidders earns at most $nK\le z$, while any price above $K$ serves only the high bidder and earns at most $z$. Hence $\F(v)=z$ and $z$ is an optimal price. This includes the boundary $z=nK$, where $z$ merely ties with the prices that serve all $n$ bidders; \eqref{eq:overestimate-guarantee} is required at every optimal price, so ties are harmless. Since the prediction is $H$, the log-error of this strict overestimate is $\log(H/z)$. Applying \eqref{eq:overestimate-guarantee} and accounting for the $n$ possible choices of the high bidder, these profiles contribute at least
\[
n\ell_K^{n-1}\int_{nK}^{H}\rho\!\left(\log\frac{H}{z}\right)\frac{dz}{z}=n\ell_K^{n-1}\int_0^{\log K}\rho(s)\,ds.
\]
This change of variables formalizes the intuition sketched previously: the revenue guarantese contributes a factor $z$ and the equal-revenue density contributes $z^{-2}$, so their product is $dz/z$, which is exactly $ds$ in logarithmic error. 

\paragraph{Consistency.}
Suppose one coordinate equals $H$ and the remaining coordinates lie in $[1,K]$. Since $H=nK^2>nK$, posting $H$ earns $H$ and is optimal. The prediction is correct, and the atom contributes at least
\[
nc\ell_K^{n-1}.
\]

Combining the three lower bounds with \eqref{eq:three-band-upper},
\[
r\int_{[1,K]^n}\Ftwo(v)\prod_{i=1}^n\frac{dv_i}{v_i^2}+n\ell_K^{n-1}\int_0^{\log K}\rho(s)\,ds+nc\ell_K^{n-1}\le nm_K^{n-1}.
\]
As $K\to\infty$, both $\ell_K$ and $m_K$ converge to one. By monotone convergence and \eqref{eq:equal-revenue-identity}, the first term converges to $rn\lambda_n$ and the second to $n\int_0^\infty\rho(s)\,ds$. Dividing by $n$ proves the theorem.
\end{proof}

An immediate consequence is that a mechanism that exhausts its budget on perfect consistency and robustness cannot retain any protection against strict overestimation. This is formalized in the following corollary:

\begin{corollary}[Performance near a correct prediction]
\label{cor:endpoint-brittle}
Suppose a mechanism is $c$-consistent and $r$-robust, with $c+\lambda_n r=1$. There are no constants $a>0$ and $S>0$ for which
\[
\E[\Rev(M(v,e^sp^*))]\ge a\F(v)
\]
holds for every profile, every optimal price $p^*$, and every $0<s\le S$.
\end{corollary}

\begin{proof}
Such a guarantee would allow $\rho(s)=a\ind\{0<s\le S\}$ in \cref{thm:unknown-error-lb}, giving $1+aS\le1$, a contradiction.
\end{proof}

The brittleness of posting $\hat p$ exactly is therefore not specific to our particular auction. Also, note that the statement concerns uniform error tolerance over an interval of errors, so it does not state that revenue is zero at any specific error or on any individual profile.

\subsection{Randomized discounting mechanism}
\label{sec:error-construction}

We now present a random-discount mechanism that attains the guarantee prescribed by our lower bound in Theorem~\ref{thm:unknown-error-lb}. The mechanism is a direct generalization of the prediction branch of \cref{def:mixture}. Instead of always posting $\hat p$, the mechanism posts randomly discounted predictions
\[
P=\hat pe^{-u},\qquad u\ge0,
\]
and runs $\mathcal{A}_n^*$ with probability $\lambda_n r$, which guarantees $r\Ftwo(v)$ on every profile. The probability mass placed at log-discount $u$ is what allows us to extract revenue from predictions that have overestimation error at most $u$, but it earns only an $e^{-u}$ fraction of the predicted price when it does, showing that error tolerance on large scales depletes the probability budget in exactly the way the lower bound predicts. Choosing the distribution of $u$ is then a matter of reading off from the target curve $\rho$ how much mass each discount requires. Randomized discounting of this kind also appears in the error-tolerant constructions of~\citet{BalcanPrasadSandholm2023} and~\citet{LuWanZhang2024}.

\begin{theorem}[Randomized discounting is Pareto optimal]
\label{thm:functional-achievability}
Fix $c,r\in[0,1]$ and a nonincreasing, left-continuous function $\rho:(0,\infty)\to[0,1]$ with
\[
\rho_0 \coloneq \lim_{s\to0^+}\rho(s)\le c.
\]
If
\begin{equation}
c+\lambda_n r+\int_0^\infty\rho(s)\,ds\le1,
\label{eq:target-profile-budget}
\end{equation}
then there is a universally truthful mechanism, using only the prediction and the current bids, with consistency at least $c$ and robustness at least $r$. For every optimal price $p^*$ and strict overestimate $\hat p=e^sp^*$, it satisfies
\[
\E[\Rev(M(v,\hat p))]\ge r\Ftwo(v)+\rho(s)\F(v),
\]
and for every underestimate $p^*=\eta\hat p$ with $\eta\ge1$, it satisfies
\[
\E[\Rev(M(v,\hat p))]\ge r\Ftwo(v)+\frac{c}{\eta}\F(v).
\]
\end{theorem}

\begin{proof}
We construct a measure $\nu$ on $[0,\infty)$ whose mass specifies the probability assigned to each log-discount $u$. Unlike a conditional probability distribution, its total mass may be less than one, because the mechanism also runs the fallback.

Let $-d\rho$ be the positive Lebesgue--Stieltjes measure defined by
\[
(-d\rho)([a,b))=\rho(a)-\rho(b),\qquad0<a<b.
\]
For a differentiable function this measure has density $-\rho'(u)$, and a downward jump from $\rho(u)$ to $\rho(u+)$ contributes an atom of mass $\rho(u)-\rho(u+)$ at $u$. Define
\begin{equation}
\nu(du) \coloneq \rho(u)\,du+(-d\rho)(du)+(c-\rho_0)\delta_0(du),
\label{eq:stieltjes-measure}
\end{equation}
where the first two measures are supported on $(0,\infty)$. All three terms are nonnegative. Since $\rho$ is nonincreasing with finite integral, $\rho(u)\to0$ as $u\to\infty$, so $(-d\rho)((0,\infty))=\rho_0$ and
\[
\nu([0,\infty))=c+\int_0^\infty\rho(u)\,du\le1-\lambda_n r.
\]

The mechanism runs $\mathcal{A}_n^*$ with probability $\lambda_n r$. For probability mass $\nu(du)$ it posts $\hat pe^{-u}$ to every bidder, and any remaining probability is assigned to no sale. These choices are independent of the current bids and the fallback is universally truthful, so the resulting mechanism is universally truthful.

We next verify the revenue guarantees. Stieltjes integration by parts gives, for $s>0$,
\[
\int_{[s,\infty)}e^{-u}(-d\rho)(du)=e^{-s}\rho(s)-\int_s^\infty e^{-u}\rho(u)\,du,
\]
and therefore
\begin{equation}
\int_{[s,\infty)}e^{-u}\nu(du)=e^{-s}\rho(s).
\label{eq:stieltjes-identities}
\end{equation}
Taking $s\to0^+$ and adding the atom at zero gives $\int_{[0,\infty)}e^{-u}\nu(du)=c$.

Fix an optimal price $p^*$ and let $k^* \coloneq N_v(p^*)$, so that $\F(v)=k^*p^*$. If $\hat p=e^sp^*$, every price with $u\ge s$ is at most $p^*$ and is accepted by at least these $k^*$ bidders, so the discounted prices contribute at least
\[
k^*\hat p\int_{[s,\infty)}e^{-u}\nu(du)=k^*\hat pe^{-s}\rho(s)=\rho(s)\F(v).
\]
If $p^*=\eta\hat p$ with $\eta\ge1$, every discounted price is at most $p^*$, and their contribution is at least
\[
k^*\hat p\int_{[0,\infty)}e^{-u}\nu(du)=\frac{c}{\eta}\F(v).
\]
The fallback contributes at least $r\Ftwo(v)$ in every case, and at $\eta=1$ we obtain consistency at least $c$.
\end{proof}

The construction accounts for the three terms of \eqref{eq:functional-frontier} one at a time. The fallback requires probability $\lambda_n r$, the discounted prices require total probability $c+\int_0^\infty\rho(s)\,ds$, and the two together cannot exceed one. The regularity assumptions on $\rho$ are mild in the following sense. A guarantee is only meaningful over a range of errors, and any function guaranteed by a mechanism may be replaced by the largest nonincreasing left-continuous function below it without changing what the mechanism delivers over such a range.

Together, \cref{thm:unknown-error-lb} and \cref{thm:functional-achievability} turn the frontier \eqref{eq:n-frontier} into an exact characterization for any nonincreasing left-continuous error tolerance profile $\rho$ with $\rho_0\le c$. A triple consisting of consistency $c$, robustness $r$, and a nonincreasing overestimation guarantee $\rho$ is achievable if and only if
\[
c+\lambda_n r+\int_0^\infty\rho(s)\,ds\le1.
\]
The mechanism attaining it additionally guarantees $(c/\eta)\F(v)$ under underestimation.

\subsection{Example applications}
\label{sec:error-examples}

We now apply our guarantee to (1) derive same-rate error degradation in both directions (that is, for overestimates and underestimates), (2) obtain tunable error guarantees, and (3) characterize the slowest possible error decay.

\paragraph{Uniform discounts and symmetric error.}
The simplest error-tolerant mechanism draws its price uniformly below the prediction. It guarantees a fraction of $\F$ proportional to $1/\eta$ in both directions, and no mechanism with the same robustness can give a better revenue guarantee.

\begin{corollary}
\label{cor:optimal-scale-free}
Suppose a universally truthful bid-independent mechanism is $r$-robust and guarantees
\[
\E[\Rev(M(v,\hat p))]\ge\frac{a}{\eta(\hat p,p^*)} \cdot \F(v)
\]
for every profile, prediction, and optimal price $p^*$. Then
\begin{equation}
a\le\frac{1-\lambda_n r}{2}.
\label{eq:symmetric-optimum}
\end{equation}
For every $r\in[0,1/\lambda_n]$ the bound is attained by running $\mathcal{A}_n^*$ with probability $\lambda_n r$ and otherwise drawing $P$ uniformly from $(0,\hat p]$, which satisfies
\begin{equation}
\E[\Rev]\ge r \cdot \Ftwo(v)+\frac{1-\lambda_n r}{2\eta(\hat p,p^*)} \cdot \F(v).
\label{eq:optimal-symmetric-guarantee}
\end{equation}
\end{corollary}

\begin{proof}
At $\eta=1$ the assumed guarantee gives consistency at least $a$, and for strict overestimates we may take $\rho(s)=ae^{-s}$. \cref{thm:unknown-error-lb} therefore gives
\[
1\ge a+\lambda_n r+\int_0^\infty ae^{-s}\,ds=2a+\lambda_n r.
\]
For achievability, write $q \coloneq 1-\lambda_n r$ and $k^* \coloneq N_v(p^*)$. If $\hat p\le p^*$, every sampled price is accepted by at least $k^*$ bidders, so the prediction branch earns at least
\[
qk^*\frac{\hat p}{2}=\frac{q}{2\eta}\F(v).
\]
If $\hat p>p^*$, counting only the sampled prices that are at most $p^*$ gives
\[
\frac{qk^*}{\hat p}\int_0^{p^*}p\,dp=\frac{qk^*(p^*)^2}{2\hat p}=\frac{q}{2\eta}\F(v).
\]
Adding the fallback contribution proves the result.
\end{proof}

The factor $1/2$ has a direct interpretation. A uniform discount earns half the predicted price per guaranteed buyer on average, and if it is an overestimate, only the part of the price range lying below $p^*$ counts toward the guarantee. The lower bound says that this is the best exchange rate available, so no mechanism improves the coefficient of a symmetric $1/\eta$ guarantee without weakening robustness.

\paragraph{Tuning the rate of degradation.}
Uniform discounting fixes the rate at which the guarantee decays. More generally, say we would like to achieve $a/\eta^b$ under overestimation for an exponent $b>0$. In logarithmic coordinates this is $\rho(s)=ae^{-bs}$, whose integral is $a/b$, so \cref{thm:unknown-error-lb} says that any $c$-consistent, $r$-robust mechanism with $0\le a\le c$ satisfies
\begin{equation}
c+\frac ab+\lambda_n r\le1.
\label{eq:power-law-frontier}
\end{equation}
Conversely, whenever \eqref{eq:power-law-frontier} holds, \cref{thm:functional-achievability} attains these guarantees using the log-discount measure $\nu(ds)=a(1+b)e^{-bs}\,ds+(c-a)\delta_0(ds)$, and simultaneously guarantees $(c/\eta)\F(v)$ under underestimation.

We let $a=c$. When $a<c$, the mechanism posts $\hat p$ itself with probability $c-a$, which earns nothing once $\hat p>p^*$, so the guarantee drops from $c$ to $a$ even under the smallest overestimate. With $a=c$, $\hat p$ is never offered, so the guarantee is continuous in $\eta$ at a correct prediction. Then, fix $r\in[0,1/\lambda_n]$ and write $q \coloneq 1-\lambda_n r$. The largest consistency in this family is
\begin{equation}
c_b^* \coloneq q\frac{b}{b+1},
\label{eq:power-law-corner}
\end{equation}
achieved by running the fallback with probability $\lambda_n r$ and otherwise posting $Z\hat p$, where $\Pr(Z\le z)=z^b$ for $0\le z\le1$, that is, $Z\sim\mathrm{Beta}(b,1)$. The resulting guarantee is
\[
\E[\Rev]\ge r\Ftwo(v)+c_b^*
\begin{cases}
\eta^{-1}\F(v),&\hat p\le p^*,\\[1mm]
\eta^{-b}\F(v),&\hat p>p^*.
\end{cases}
\]
Larger $b$ places more probability on discounts close to one, which improves consistency but makes the overestimation guarantee decay faster; smaller $b$ protects against larger overestimates at the cost of consistency; and $b=1$ recovers uniform discounting. Writing $s=\log(\hat p/p^*)$ for the signed logarithmic error and
\[
\Phi_b(s) \coloneq c_b^*
\begin{cases}
e^s,&s\le0,\\
e^{-bs},&s>0
\end{cases}
\]
for the coefficient of $\F(v)$ in the display above, Figure~\ref{fig:power-law} compares three choices of $b$.

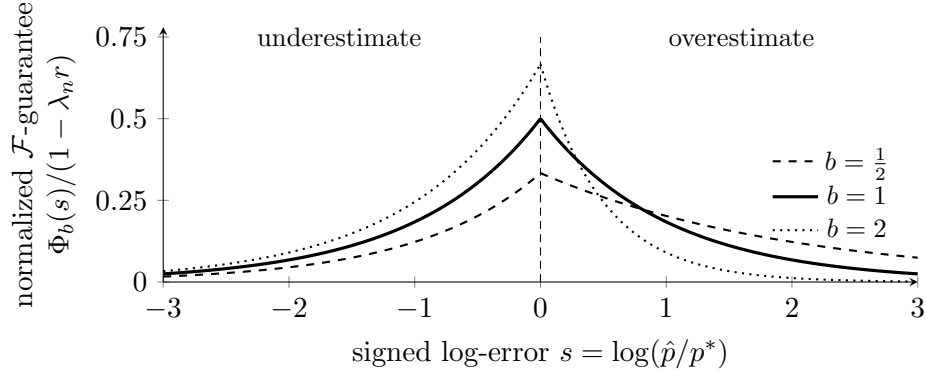
\begin{figure}[H]
    \centering
    \begin{tikzpicture}
        \begin{axis}[
            width=0.70\linewidth,
            height=0.30\linewidth,
            xmin=-3, xmax=3,
            ymin=0, ymax=0.78,
            axis lines=left,
            xlabel={signed log-error $s=\log(\hat p/p^*)$},
            ylabel style={align=center}, 
            ylabel={normalized $\F$-guarantee\\ $\Phi_b(s)/(1-\lambda_n r)$},
            xtick={-3,-2,-1,0,1,2,3},
            ytick={0,0.25,0.5,0.75},
            legend style={at={(0.98,0.57)},anchor=north east,draw=none,fill=none,font=\small},
            samples=120,
            clip=false,
        ]
        \addplot[thick,dashed,domain=-3:0] {(1/3)*exp(x)};
        \addplot[thick,dashed,domain=0:3,forget plot] {(1/3)*exp(-0.5*x)};
        \addplot[very thick,domain=-3:0] {0.5*exp(x)};
        \addplot[very thick,domain=0:3,forget plot] {0.5*exp(-x)};
        \addplot[thick,dotted,domain=-3:0] {(2/3)*exp(x)};
        \addplot[thick,dotted,domain=0:3,forget plot] {(2/3)*exp(-2*x)};
        \legend{$b=\tfrac12$,$b=1$,$b=2$}
        \draw[densely dashed] (axis cs:0,0)--(axis cs:0,0.75);
        \node[anchor=south,font=\small] at (axis cs:-1.6,0.69) {underestimate};
        \node[anchor=south,font=\small] at (axis cs:1.6,0.69) {overestimate};
        \end{axis}
    \end{tikzpicture}
    \caption{Revenue guarantees obtained by posting $Z\hat p$ with probability $q=1-\lambda_n r$, where $\Pr(Z\le z)=z^b$. We take $r<1/\lambda_n$ and divide the coefficient of $\F(v)$ by $q$. Larger $b$ improves the guarantee at a correct prediction but causes faster degradation under overestimation. The fallback contribution $r\Ftwo(v)$ is not shown.}
    \label{fig:power-law}
\end{figure}

\paragraph{How slowly can an error guarantee decay?}
Since the budget is an integral, it also limits how slowly a guarantee can decay. No mechanism in our class can guarantee $C\F(v)/(1+\log\eta)$ for a fixed $C>0$ at all sufficiently large overestimation factors $\eta=\hat p/p^*$, because that would force $\rho(s)\ge C/(1+s)$ for large $s$, whose integral diverges. Any nonincreasing integrable profile (satisfying the mild conditions of Theorem~\ref{thm:functional-achievability}) that decays faster is available. For instance, fix $r\in[0,1/\lambda_n]$, write $q \coloneq 1-\lambda_n r$, choose $\theta>1$, and set
\[
c \coloneq q\frac{\theta-1}{\theta},\qquad\rho(s) \coloneq \frac{c}{(1+s)^\theta}.
\]
Then $c+\lambda_n r+\int_0^\infty\rho(s)\,ds=1$, so \cref{thm:functional-achievability} gives a mechanism guaranteeing
\[
\E[\Rev]\ge r\Ftwo(v)+\frac{c}{(1+\log\eta)^\theta}\F(v)
\]
under overestimation. As $\theta$ decreases to one the guarantee decays more slowly with error, but its coefficient $c$ vanishes. This is the section's central tradeoff. Tolerance against a wider range of errors costs performance when the prediction is perfect, and a single point prediction can only support guarantees whose area under the curve in log-error is finite.

\section{Learning price predictions from historical data}
\label{sec:learning-prices}

The preceding sections treat the predicted price $\hat p$ as externally supplied.
We now study how such advice can be learned from historical markets (full
proofs and further discussion appear in Appendix~\ref{sec:learning}).
We consider two settings. First, the seller learns a single \emph{shared price}
to use across future markets. Second, given public market features, the seller
learns a \emph{contextual price rule}. In both cases, the learning objective can
be chosen to optimize the downstream auction revenue directly.\looseness-1

\paragraph{Learning a shared price.}
Suppose valuation profiles $V\in[0,H]^n$ are drawn i.i.d.\ from an unknown
distribution $\mathcal D$. For $p\geq 0$, define
\[
    u_p(v)
     \coloneq  \Rev(v;p)
    = pN_v(p)
    = p\,\bigl|\{i\in[n]:v_i\geq p\}\bigr|,
    \qquad
    \mathcal U_n \coloneq \{u_p:p\geq 0\}.
\]
Thus, rather than predicting the ex post optimal price separately for each
realized market, the learner seeks the best single price
\[
    p^*_{\mathcal D}
    \in \arg\max_{p\geq0}
    \E_{V\sim\mathcal D}[u_p(V)].
\]
Despite being parameterized by only one real number, this revenue class has
statistical complexity that grows logarithmically with the number of bidders.

\begin{theorem}[Pseudo-dimension of uniform-price revenue]
\label{thm:shared-price-pdim}
There is a universal constant $C>0$ such that, for every $n\geq4$,
\[
    \left\lfloor \log_2\frac n2\right\rfloor
    \leq
    \Pdim(\mathcal U_n)
    \leq
    C\log n.
\]
Hence $\Pdim(\mathcal U_n)=\Theta(\log n)$.
The lower bound holds even when all valuations lie in $[0,1]^n$.
\end{theorem}

The upper bound can also be viewed through the piecewise-decomposable framework
of~\citet{balcan2021learning}. On every fixed market, revenue as a function of $p$
has only $O(n)$ pieces. The matching lower bound shows that the resulting
logarithmic dependence on market size is unavoidable.

Standard pseudo-dimension generalization bounds~\citep{anthony_bartlett} immediately give the following
sample-complexity guarantee. If
\[
    \hat p\in
    \arg\max_{p\geq0}
    \frac1T\sum_{t=1}^T u_p(V^{(t)})
\]
is an empirical revenue-maximizing price, then
\begin{corollary}[Learning a near-optimal shared price]
\label{cor:shared-price-learning-main}
For $\epsilon,\delta\in(0,1)$,
\[
    T
    =
    O\!\left(
        \frac{\log n+\log(1/\delta)}{\epsilon^2}
      \right)
\]
samples suffice so that, with probability at least $1-\delta$,
\[
    \E[u_{\hat p}(V)]
    \geq
    \sup_{p\geq0}\E[u_p(V)]
    -\epsilon nH.
\]
\end{corollary}

\noindent Moreover,
\[
    \sup_{p\geq0}\E[u_p(V)]
    \;\leq\;
    \E\!\left[\sup_{p\geq0}u_p(V)\right]
    =
    \E[F(V)].
\]
Thus, learning the best shared price and achieving instance-wise consistency
are distinct objectives, and $\hat p$ need not be an optimal price for each test
market.
Nevertheless, the statistical result composes directly with our robust
mechanism. If $A$ is any prediction-free mechanism satisfying
$\E[\Rev(A(v))]\geq F^{(2)}(v)/\lambda$ and, on a new market, we offer $\hat p$
with probability $\alpha$ and run $A$ otherwise, then with the same high
probability over the training sample,
\[
\begin{split}
    \E_{V,\mathrm{mech}}[\Rev(M_{\alpha,A}(V,\hat p))]
    \geq\;&
    \alpha\sup_{p\geq0}\E[u_p(V)]
    +
    \frac{1-\alpha}{\lambda}\E[F^{(2)}(V)]
    -\alpha\epsilon nH.
\end{split}
\]
In particular, taking $A=\mathcal{A}_n^*$ gives the guarantee with the optimal
prior-free fallback.

\paragraph{Learning contextual prices.}
When public covariates $X\in\mathcal X$ are available before bids are
collected, let $\mathcal H\subseteq\{h:\mathcal X\to\mathbb R_{\geq0}\}$ be a class of market-specific price predictors with \looseness-1
\[
    d_{\mathcal H} \coloneq \Pdim(\mathcal H).
\]
One possibility is \emph{prediction-focused} learning, in which one predicts a
selected ex post optimal price from $X$. Alternatively, and more directly
aligned with the seller's objective, one can optimize revenue itself. Define
\[
    u_h(x,v) \coloneq h(x)N_v(h(x)),
    \qquad
    \mathcal U_{\mathcal H}
     \coloneq \{u_h:h\in\mathcal H\},
\]
and choose $\hat h$ by empirical revenue maximization.

\begin{theorem}[Contextual price learning]
\label{thm:contextual-price-learning}
If $d_{\mathcal H}\geq1$, then
\[
    \Pdim(\mathcal U_{\mathcal H})
    =
    O(d_{\mathcal H}\log n).
\]
Consequently, when bidder values lie in $[0,H]$,
\[
    T
    =
    O\!\left(
       \frac{d_{\mathcal H}\log n+\log(1/\delta)}
            {\epsilon^2}
      \right)
\]
samples suffice for contextual ERM to return $\hat h$ satisfying, with
probability at least $1-\delta$,
\[
    \E[u_{\hat h}(X,V)]
    \geq
    \sup_{h\in\mathcal H}\E[u_h(X,V)]
    -\epsilon nH.
\]
\end{theorem}

\noindent For constant predictors, $d_{\mathcal H}=1$, recovering the
$O(\log n)$ shared-price upper bound. So far, the result characterizes statistical learnability but does not say anything about computational efficiency. For the constant-price class, exact ERM is computationally efficient. This is also the case for certain simple low-dimensional predictor classes. However, ERM can be computationally hard for general $\mathcal H$. Finally, all learning must use only historical data and
public covariates available before the current bids are observed. This
exogeneity ensures that the learned pricing rule remains independent of each
bidder's reported bid and hence is compatible with the truthfulness
guarantees developed above.\looseness-1

\section{Conclusion}\label{sec:conclusion}

We present a model of unlimited supply auction design with price predictions, where the seller receives a single scalar prediction of the optimal monopoly price. This is in contrast from the majority of prior work on auctions with predictions which typically assumes predictions of every bidder's value. We argue that a population-level price prediction is more likely to be available in practice and is a more accurate reflection of the current practice of data-driven revenue and price optimization.

For unlimited-supply auctions with a scalar prediction of the optimal uniform price, we show that the optimal
consistency-robustness tradeoff has a simple form: If $\lambda_n$ is the exact prior-free ratio
against $\Ftwo$, then every universally truthful bid-independent mechanism satisfies
$c+\lambda_n r\le1$, and equality is achieved by randomizing between the predicted price and an optimal
prior-free auction. The budgeting lemma and resulting lower bound show an inherent tradeoff between robustness and consistency: robustness requires the pricing rule to utilize many lower values when the prediction is too high, while consistency requires the rule to place probability mass at the predicted high value. 

The budgeting perspective also applies to imperfect predictions. Without a known error radius, \cref{thm:unknown-error-lb} provides an exact integral budget for the entire error-tolerance curve.
\cref{sec:learning} outlines how these predictions can be produced from historical data, with shared-price revenue having pseudo-dimension $\Theta(\log n)$ and contextual revenue having complexity $O(d_{\mathcal H}\log n)$. Together, these results provide a complete path from statistically generated price predictions to truthful use of those predictions and worst-case revenue protection. Our work leaves open several interesting questions:

\begin{enumerate}[leftmargin=*]
    \item \textbf{Truthfulness in expectation.}
    Our mechanism class is universally truthful, matching the bid-independent posted-price model used in the classical competitive-auction characterization. It would be useful to determine whether allowing the broader class of truthful-in-expectation mechanisms can alter/expand the Pareto frontier.\looseness-1

    \item \textbf{Efficient optimal fallbacks.}
    \citet{ChenGravinLu2014} prove exact attainability of $\lambda_n$, but the general construction is considerably less explicit than RSOP or SCS. Finding a succinct, computationally practical auction attaining or nearly attaining $\lambda_n$ would immediately improve the practicality of our randomized mechanism.

    \item \textbf{Beyond unlimited supply of a single item.}
    The natural robustness benchmark in limited-supply and downward-closed permutation environments is $\mathrm{EFO}^{(2)}$.  Extending the budgeting argument to feasibility-coupled allocations may require combining product equal-revenue measures with benchmark decomposition \citep{HartlineYan2011,ChenGravinLu2015}. Extending our results to handle the sale of multiple items is an important direction as well.
       \item \textbf{Efficient data-driven pricing.}
    Our learning results characterize the statistical complexity of optimizing shared and contextual prices directly for revenue, but efficient empirical revenue maximization is understood only for relatively simple predictor classes. Developing computationally efficient algorithms for richer contextual pricing classes, while retaining generalization guarantees and exogeneity from current private bids, is an important direction.
\end{enumerate}

\subsection*{Acknowledgments}

This work was supported in part by the National Science Foundation under grants ECCS-2216899,
ECCS-2216970, and by an NSF Graduate Research Fellowship.

\subsection*{AI Disclosure}
AI tools were used to assist with literature search as well as to draft and edit portions of the paper. ChatGPT 5.6 assisted in the simple extrapolation of the budgeting lemma argument for 2-bidders (\cref{thm:two-frontier}) to $n$ bidders (\cref{thm:n-frontier}) as well as deriving the error-tolerant generalization of the consistency-robustness Pareto frontier. The authors take full responsibility for all content.

\newpage

\bibliographystyle{plainnat}
\bibliography{references}

\newpage

\appendix

\section{Learning price predictions from historical data}
\label{sec:learning}

The preceding sections treat the price prediction as externally supplied.
We now ask how such advice can be learned from historical markets. We study two
natural, but distinct, decision problems. Our goal is to learn a uniform price for new markets, but this price may be shared across markets or be market-specific as described below.\looseness-1

The first learns a single {\it shared} price that performs well in expected revenue across
future markets. This objective need not predict the ex post optimal price of
each individual market. We show that this problem has low statistical
complexity. Although the class is indexed by only a single real-valued price
parameter, its pseudo-dimension is $\Theta(\log n)$. Thus the complexity is
driven by the size of the market rather than the dimension of the parameter
space.

When public covariates are available, the seller can instead learn a
market-specific price. We consider both predicting a selected ex post optimal
price and training the contextual price rule directly for revenue. The latter
objective yields an $O(d_{\mathcal H}\log n)$ pseudo-dimension bound for a
predictor class $\mathcal H$ of pseudo-dimension $d_{\mathcal H}$.\looseness-1

\subsection{Learning a shared price by empirical revenue maximization}
\label{sec:shared-price-learning}

Suppose the seller repeatedly encounters markets with $n$ bidders, whose
valuation profiles are drawn i.i.d.\ from an unknown distribution $\mathcal D$.
For $p\ge 0$, define
\[
u_p(v)
 \coloneq 
\Rev(v;p)
=
p \cdot N_v(p)
=
p\left|\{i\in[n]:v_i\ge p\}\right|,
\]
and let
\[
\mathcal U_n \coloneq \{u_p:p\ge0\}.
\]
Thus $u_p(v)$ is exactly the revenue obtained by posting the common price $p$
to all bidders. On every fixed valuation profile,
\[
\sup_{p\ge0}u_p(v)=\F(v).
\]
Here, however, the learner does not optimize separately on each realized
profile. Rather, it seeks the best single price
\[
p_{\mathcal D}^*
\in
\arg\max_{p\ge0}
\E_{V\sim\mathcal D}[u_p(V)]
\]
to use across markets drawn from $\mathcal D$.

We first characterize the statistical complexity of this learning problem.
The result is somewhat surprising since the class has only one real-valued
parameter, but its pseudo-dimension grows logarithmically with the number of
bidders. \cref{thm:price-pdim} gives a tight characterization for the particularly
simple class of single uniform posted prices. This complements the broader use
of pseudo-dimension to study learnability of auction classes by \citet{MorgensternRoughgarden2015}. Here the class is
one-dimensional, the dependence on the number of bidders is characterized in
both directions, and the lower bound already holds on the bounded domain
$[0,1]^n$.

\begin{theorem}[Pseudo-dimension of uniform-price revenue]
\label{thm:price-pdim}
There is a universal constant $C>0$ such that, for every $n\ge4$,
\[
\left\lfloor\log_2\frac n2\right\rfloor
\le
\Pdim(\mathcal U_n)
\le
C\log n.
\]
Consequently,
$\Pdim(\mathcal U_n)=\Theta(\log n)$.
The lower bound continues to hold even when all bidder values are restricted
to $[0,1]$.
\end{theorem}

\begin{proof}

    \emph{Upper bound.}
Fix valuation profiles
$v^{(1)},\ldots,v^{(m)}$
and witness thresholds $r_1,\dots,r_m$ for pseudo-shattering. Since every
$u_p(v)$ is nonnegative, any coordinate with $r_j<0$ is labeled $1$ for every
p and therefore cannot be pseudo-shattered. Hence we may assume
$r_j\ge0$ for every $j$.

For a fixed profile $v$, write its order statistics as
\[
v_{(1)}\ge v_{(2)}\ge\cdots\ge v_{(n)},
\qquad
v_{(n+1)} \coloneq 0.
\]
Viewed as a function of the parameter $p$,
\[
u_v^*(p) \coloneq u_p(v)=p \cdot N_v(p)
\]
satisfies
\[
u_v^*(p)=kp
\qquad
\text{for }p\in(v_{(k+1)},v_{(k)}],
\]
and $u_v^*(p)=0$ for $p>v_{(1)}$. Thus, it has at most $n+1$
linear pieces.

For a fixed threshold $r\ge0$, on the $k$th nonzero piece the inequality
$u_p(v)>r$
is simply
$kp>r$.
Consequently, the set
\[
S(v,r) \coloneq \{p\ge0:u_p(v)>r\}
\]
is a union of at most $n$ intervals and therefore has at most $2n$
endpoints.

Apply this observation to each pair $(v^{(j)},r_j)$. Across all $m$
instances there are at most $2mn$ endpoints on the real line. These endpoints
and the open intervals between consecutive endpoints form $O(mn)$ cells, and
the complete labeling
\[
\left(
\mathbf 1\{u_p(v^{(1)})>r_1\},
\ldots,
\mathbf 1\{u_p(v^{(m)})>r_m\}
\right)
\]
is constant on each such cell. Hence the number of labelings realized by
varying the single price $p$ is at most $O(mn)$.

If the $m$ profiles were pseudo-shattered, all $2^m$ labelings would have to
be realized. Therefore
\[
2^m\le O(mn).
\]
This implies
$m=O(\log n)$,
and hence
\[
\Pdim(\mathcal U_n)=O(\log n).
\]

\emph{Lower bound.}
Let
\[
m \coloneq \left\lfloor\log_2\frac n2\right\rfloor,
\qquad
L \coloneq 2^m,
\]
so that $2L\le n$. Enumerate all $L=2^m$ binary vectors as
\[
b^{(1)},\ldots,b^{(L)}\in\{0,1\}^m.
\]
We construct $m$ valuation profiles
\[
v^{(1)},\ldots,v^{(m)}\in[0,1]^n
\]
that are pseudo-shattered by $\mathcal U_n$, using the common witness threshold
\[
r_1=\cdots=r_m=1.
\]

For $\ell\in[L]$, define
\[
t_\ell \coloneq 2(L-\ell+1),
\qquad
p_\ell \coloneq \frac{1}{t_\ell-\frac12}.
\]
Since
\[
t_1>t_2>\cdots>t_L\ge2,
\]
we have
\[
0<p_1<p_2<\cdots<p_L=\frac23.
\]

\noindent For each instance $j\in[m]$ and each $\ell\in[L]$, prescribe the number
of bidders whose values are at least $p_\ell$ to be\looseness-1
\[
c_{j,\ell}
 \coloneq 
\begin{cases}
t_\ell, & b_j^{(\ell)}=1,\\
t_\ell-1, & b_j^{(\ell)}=0.
\end{cases}
\]
These counts form a valid nonincreasing demand curve. Indeed, since
$t_{\ell+1}=t_\ell-2$, regardless of the two adjacent bits,
\[
c_{j,\ell}
\ge
t_\ell-1
>
t_\ell-2
\ge
c_{j,\ell+1},
\]
and
\[
c_{j,1}\le t_1=2L\le n.
\]
To realize these counts, set $c_{j,L+1} \coloneq 0$. For each $\ell\in[L]$, give
\[
c_{j,\ell}-c_{j,\ell+1}
\]
bidders value exactly $p_\ell$, and give each of the remaining
$n-c_{j,1}$ bidders value $p_1/2$. Then, by telescoping,
\[
\left|\{i:v_i^{(j)}\ge p_\ell\}\right|
=
c_{j,\ell},
\]
and all valuations lie in $(0,2/3]\subseteq[0,1]$.
At price $p_\ell$, if $b_j^{(\ell)}=1$, then
\[
u_{p_\ell}(v^{(j)})
=
p_\ell t_\ell
=
\frac{t_\ell}{t_\ell-\frac12}
>
1,
\]
whereas if $b_j^{(\ell)}=0$, then
\[
u_{p_\ell}(v^{(j)})
=
p_\ell(t_\ell-1)
=
\frac{t_\ell-1}{t_\ell-\frac12}
<
1.
\]
Thus
\[
u_{p_\ell}(v^{(j)})>r_j
\quad\Longleftrightarrow\quad
b_j^{(\ell)}=1.
\]
Because $b^{(1)},\ldots,b^{(L)}$ enumerate all $2^m$ labelings of the
$m$ instances, the price $p_\ell$ realizes labeling $b^{(\ell)}$. Hence
$\{v^{(1)},\ldots,v^{(m)}\}$ is pseudo-shattered and
\[
\Pdim(\mathcal U_n)
\ge
m
=
\left\lfloor\log_2\frac n2\right\rfloor.
\]
The separation is strict in both directions, so the construction works under
either the strict or non-strict convention for pseudo-shattering.
\end{proof}

\begin{remark}[Relation to piecewise-decomposable bounds]
\label{rem:piecewise-price-pdim}
The upper bound can alternatively be obtained from the
piecewise-decomposable framework of \citet{balcan_piecewise_decomposable}. For each fixed profile, the dual revenue
function has at most $n+1$ linear pieces, the boundary functions are
one-dimensional thresholds and the piece functions are linear maps
$p\mapsto kp$. The elementary proof above avoids any primal-dual bookkeeping
and makes transparent that the logarithmic dependence arises simply because
an arrangement of $m$ unions of $O(n)$ intervals on a one-dimensional
parameter space has only $O(mn)$ cells.
\end{remark}

\noindent \cref{thm:price-pdim} immediately gives a sample-complexity guarantee.

\begin{corollary}[Learning a near-optimal shared price]
\label{cor:shared-price-learning}
Suppose bidder values lie in $[0,H]$, and let
\[
V^{(1)},\ldots,V^{(T)}
\stackrel{\mathrm{i.i.d.}}{\sim}\mathcal D.
\]
Let
\[
\widehat p
\in
\arg\max_{p\ge0}
\frac1T
\sum_{t=1}^T
u_p(V^{(t)})
\]
be an empirical revenue-maximizing price. For \(\epsilon,\delta\in(0,1)\), there is a universal constant
$C>0$ such that
\[
T
\ge
C \cdot 
\frac{\log n+\log(1/\delta)}{\epsilon^2}
\]
suffices to guarantee, with probability at least $1-\delta$,
\[
\E_{V\sim\mathcal D}[u_{\widehat p}(V)]
\ge
\sup_{p\ge0}
\E_{V\sim\mathcal D}[u_p(V)]
-
\epsilon nH.
\]
\end{corollary}
\begin{proof}
Normalize revenue by defining
\[
\bar u_p(v) \coloneq \frac{u_p(v)}{nH}\in[0,1].
\]
Rescaling does not change pseudo-dimension. Standard uniform-convergence
bounds for bounded real-valued function classes in terms of pseudo-dimension
\citet{anthony_bartlett}, together with \cref{thm:price-pdim} and the
usual ERM argument, give the result.
\end{proof}

\begin{remark}[The additive error scale]
\label{rem:additive-learning-scale}
The error term in \cref{cor:shared-price-learning} is additive on
the scale $nH$, the largest possible revenue under the bounded-value
assumption. It is therefore not a relative approximation to
\[
\sup_{p\ge0}\E[u_p(V)].
\]
If only a small number of bidders typically purchase at the best shared
price, this population optimum may be much smaller than $nH$, in which case
the bound can be weak or vacuous as a relative guarantee. Obtaining
distribution-dependent relative guarantees, for example through localized
complexity or variance-sensitive arguments, is an interesting refinement.
\end{remark}

The statistical guarantee composes immediately with the robust mechanism
developed earlier. Moreover, because $\widehat p$ is learned from an
independent historical sample before the current bids are collected, it is
automatically exogenous to the current bidders' reports.

\begin{corollary}[End-to-end revenue guarantee]
\label{cor:end-to-end-learning}
Let $A$ be any prediction-free universally truthful mechanism satisfying
\[
\E[\Rev(A(v))]
\ge
\frac{F^{(2)}(v)}{\lambda}
\qquad
\text{for every }v.
\]
Train $\widehat p$ as in \cref{cor:shared-price-learning} on an
independent historical sample. On a new market, run the trust-$\alpha$
mixture: with probability $\alpha$, offer $\widehat p$ to every bidder, and
with probability $1-\alpha$, run $A$.

Then, with probability at least $1-\delta$ over the historical sample,
\[
\begin{aligned}
\E_{V\sim\mathcal D,\mathrm{mech}}
\left[
\Rev(M_{\alpha,A}(V,\widehat p))
\right]
\ge\;&
\alpha
\sup_{p\ge0}
\E_{V\sim\mathcal D}[u_p(V)]
\\
&+
\frac{1-\alpha}{\lambda}
\E_{V\sim\mathcal D}[F^{(2)}(V)]
-
\alpha\epsilon nH.
\end{aligned}
\]
In particular, taking $A=A_n^*$ and $\lambda=\lambda_n$ gives the
corresponding guarantee with the optimal prior-free fallback.
\end{corollary}

\begin{proof}
Conditional on the learned price $\widehat p$,
\[
\begin{aligned}
\E_{V,\mathrm{mech}}
\left[
\Rev(M_{\alpha,A}(V,\widehat p))
\right]
&=
\alpha\E_V[u_{\widehat p}(V)]
+
(1-\alpha)\E_V[\Rev(A(V))]
\\
&\ge
\alpha\E_V[u_{\widehat p}(V)]
+
\frac{1-\alpha}{\lambda}
\E_V[F^{(2)}(V)].
\end{aligned}
\]
Applying \cref{cor:shared-price-learning} to the first term
completes the proof.
\end{proof}

\begin{remark}[Shared-price learning versus instance-wise consistency]
\label{rem:shared-price-not-consistency}
The learned shared price $\widehat p$ need not equal an optimal price for an
individual test market. The two objectives are related by the elementary
inequality
\[
\underbrace{
\sup_{p\ge0}
\E_{V\sim\mathcal D}[u_p(V)]
}_{\text{best shared-price revenue}}
\;\le\;
\underbrace{
\E_{V\sim\mathcal D}
\left[
\sup_{p\ge0}u_p(V)
\right]
}_{\E[\F(V)]}.
\]
Thus decision-focused learning benchmarks against the weaker,
order-swapped quantity $\sup_p\E[u_p(V)]$, whereas consistency concerns the
instance-wise optimum $\F(V)=\sup_p u_p(V)$. Consequently,
\cref{cor:end-to-end-learning} does not invoke the consistency
guarantee from the earlier sections. It analyzes the actual revenue of the
prediction branch directly.
\end{remark}
\subsection{Contextual price learning: prediction-focused and decision-focused}
\label{sec:contextual-learning}

When informative public features are available, the seller can learn a
market-specific price rather than one common price for all future markets.
Let  $(X_1,V_1),\ldots,(X_T,V_T)$
be historical markets, where $X_t\in\mathcal X$ consists of public
covariates available before bids are collected and $V_t$ is the realized
valuation profile.  Let
\[
    \mathcal H\subseteq \{h:\mathcal X\to\mathbb R_{\ge 0}\},
    \qquad
    d_{\mathcal H} \coloneq \Pdim(\mathcal H),
\]
be a class of contextual price rules.  There are two natural (and generally
different) ways to train $h$.

\paragraph{Prediction-focused learning.}
For a valuation profile $v$, let $P^*(v) \coloneq \arg\max_{p\ge0} pN_v(p)$
denote the set of optimal uniform prices. Fix a deterministic tie-breaking
rule and select $p_t^*\in P^*(V_t)$ (assumed to be strictly positive) for each historical market, for example the
largest optimal price. Since multiplicative error is the natural scale for
prices, one may regress $Y_t \coloneq \log p_t^*$ on $X_t$\footnote{We assume here that the selected optimal price is strictly positive for technical convenience. If an all-zero market is permitted, every positive price earns zero, and there is no largest optimal price over $p\ge0$.}. A linear model, tree-based
predictor, neural network, or other regression class can then produce a
market-specific point prediction.

The tie-breaking rule is needed only to define a supervised-learning label;
the revenue benchmark itself does not require the optimal price to be unique.
This prediction-focused objective attempts to reproduce an ex post optimal
price, and therefore differs from directly optimizing the downstream revenue
of the learned price rule. \looseness-1

\paragraph{Decision-focused contextual ERM.}
Instead of trying to reproduce an ex post optimal price, one can train the
same contextual price rule directly for the quantity the seller ultimately
cares about, i.e., revenue.  Define
\[
    u_h(x,v)
       \coloneq h(x)N_v(h(x)),
    \qquad
    \mathcal U_{\mathcal H}
       \coloneq \{u_h:h\in\mathcal H\}.
\]
Given the historical sample, contextual empirical revenue maximization
chooses
\begin{equation}
    \widehat h
    \in
    \arg\max_{h\in\mathcal H}
    \widehat R_T(h),
    \qquad
    \widehat R_T(h)
       \coloneq \frac1T\sum_{t=1}^T
        h(X_t)N_{V_t}(h(X_t)).
    \label{eq:contextual-erm}
\end{equation}
This objective requires no selected optimal-price labels and does not penalize
price error for its own sake.
Two prices that are numerically far apart are treated similarly if they
produce similar revenue, while a numerically small overestimate can be
penalized heavily if it eliminates sales.  Thus the two objectives may
select different predictors even when they use the same class
$\mathcal H$.

The statistical complexity of the decision-focused objective remains
controlled by that of the underlying predictor class.

\begin{proposition}[Contextual revenue-class upper bound]
\label{prop:contextual-pdim}
Suppose $d_{\mathcal H}\ge1$. Then
\[
    \Pdim(\mathcal U_{\mathcal H})
      =O\!\left(d_{\mathcal H}\log n\right).
\]
\end{proposition}

\begin{proof}
Suppose that $m$ examples
\[
    (x_1,v^{(1)}),\ldots,(x_m,v^{(m)})
\]
are pseudo-shattered by $\mathcal U_{\mathcal H}$ with witness thresholds
$r_1,\ldots,r_m$.  As in the proof of \cref{thm:price-pdim}, we may assume
$r_j\ge0$.  For each $j$, define
\[
    S_j
       \coloneq \{p\ge0:pN_{v^{(j)}}(p)>r_j\}.
\]
The proof of \cref{thm:price-pdim} shows that $S_j$ is a union of at most $n$
intervals.  Hence membership of $h(x_j)$ in $S_j$ is determined by at most
$2n$ threshold comparisons of the form
$h(x_j)>a$.
Across all $m$ examples there are at most $M\le2mn$ such comparisons.

Since $\Pdim(\mathcal H)=d_{\mathcal H}$, the subgraph class
\[
    \{(x,a)\mapsto \ind\{h(x)>a\}:h\in\mathcal H\}
\]
has VC dimension $d_{\mathcal H}$.  Sauer's lemma therefore bounds the
number of joint outcomes of the $M$ comparisons by\looseness-1
\[
    \left(\frac{eM}{d_{\mathcal H}}\right)^{d_{\mathcal H}}
    \le
    \left(\frac{2emn}{d_{\mathcal H}}\right)^{d_{\mathcal H}},
\]
when $M\ge d_{\mathcal H}$. The remaining case is immediate.
Every labeling induced by $\mathcal U_{\mathcal H}$ is determined by these
comparison outcomes.  Thus pseudo-shattering requires
\[
    2^m
      \le
    \left(\frac{2emn}{d_{\mathcal H}}\right)^{d_{\mathcal H}}.
\]
Writing $z=m/d_{\mathcal H}$ and inverting
$2^z\le2enz$ gives $z=O(\log n)$, proving the result.
\end{proof}

For constant predictors, $d_{\mathcal H}=1$, so
\cref{prop:contextual-pdim} recovers the $O(\log n)$ upper bound
of \cref{thm:price-pdim}.  The matching lower bound there shows that the logarithmic
dependence on market size cannot in general be removed even for this
simplest predictor class.

\begin{corollary}[Contextual empirical revenue maximization]
\label{cor:contextual-erm}
Suppose all bidder values lie in $[0,H]$, and let
$(X_t,V_t)_{t=1}^T$ be drawn i.i.d.\ from a distribution $\mathcal D$.
Let $\widehat h$ be an empirical revenue maximizer as in
\eqref{eq:contextual-erm}.  For \(\epsilon,\delta\in(0,1)\), there is a universal constant $C>0$ such that
\[
    T
      \ge
    C\,
    \frac{
      d_{\mathcal H}\log n +\log(1/\delta)
    }{\epsilon^2}
\]
implies, with probability at least $1-\delta$,
\[
    \E_{(X,V)\sim\mathcal D}
      [u_{\widehat h}(X,V)]
    \ge
    \sup_{h\in\mathcal H}
      \E_{(X,V)\sim\mathcal D}[u_h(X,V)]
    -\epsilon nH.
\]
\end{corollary}

\begin{proof}
Under the bounded-value assumption, $0\le u_h(x,v)\le nH$.
Normalize by $nH$, apply the standard pseudo-dimension uniform-convergence
bound together with \cref{prop:contextual-pdim}, and use the
usual ERM argument.
\end{proof}

The guarantee also holds for approximate ERM.  If the optimization
procedure returns $\widetilde h$ whose empirical revenue is within
$\zeta$ of the empirical optimum, then the population guarantee above
incurs only the corresponding additional additive optimization error
$\zeta$ (with the uniform-convergence constants chosen accordingly).

\paragraph{How is ERM implemented?}
\cref{prop:contextual-pdim} is a statistical statement and does
not by itself imply that the optimization problem
\eqref{eq:contextual-erm} can be solved efficiently for an arbitrary
hypothesis class $\mathcal H$.  The computational problem depends on the
representation of $\mathcal H$.

For the constant-price class of \cref{sec:shared-price-learning}, exact ERM is particularly
simple.  The empirical objective is
\[
    \frac{p}{T}\sum_{t=1}^T N_{V_t}(p),
\]
which is linear between consecutive values appearing in the historical
sample.  Hence an optimum occurs at one of the at most $nT$ observed
bidder values.  Sorting these values and scanning them gives an exact ERM
in $O(nT\log(nT))$ time.

More generally, consider a $d$-dimensional affine predictor
\[
    h_\theta(x)=\langle\theta,x\rangle,
    \qquad \theta\in\Theta,
\]
where $\Theta$ is a bounded polytope and prices are constrained to be
nonnegative for every $x \in \mathcal X$.  The $nT$ hyperplanes
\[
    \langle\theta,X_t\rangle=V_{t,i},
    \qquad t\in[T],\ i\in[n],
\]
partition parameter space into cells.  Within any one cell,
$N_{V_t}(h_\theta(X_t))$ is fixed for every $t$, so the empirical revenue
objective is affine in $\theta$.  Thus exact ERM can be implemented by
enumerating the cells (including all faces of the hyperplane arrangement) and solving a linear optimization problem on each
one.  For fixed constant dimension $d$, the arrangement has
$O\left((nT)^d\right)$ cells, giving a polynomial-time exact procedure.

For richer predictor classes, exact ERM may be computationally difficult.
One may instead use a mixed-integer formulation, class-specific
optimization algorithms, or approximate empirical revenue maximization.
The statistical guarantee above should therefore be read as an
\emph{oracle-efficient} learnability result. Whenever ERM, or sufficiently
accurate approximate ERM, can be implemented for $\mathcal H$, its sample
complexity is controlled by
$O(d_{\mathcal H}\log n)$.

\paragraph{Truthfulness and exogeneity.}
Both training objectives use only historical outcomes and public covariates.
The learned price rule must be fixed independently of the current bidder
reports. In particular, model fitting, feature construction, and any ERM
computation must not use a bidder's current private report to determine the
price subsequently offered to that bidder. Frozen model parameters,
time-based sample splitting, and features available before bids are collected
provide natural ways to maintain this separation.

\end{document}